\documentclass[12pt]{iopart}
\usepackage[utf8]{inputenc}

\usepackage{amsthm}

\usepackage{amssymb}
\usepackage{amsbsy}
\usepackage{bbm}
\usepackage{enumerate}
\usepackage{graphicx}

\newtheorem{theorem}{Theorem}[section]
\newtheorem{corollary}[theorem]{Corollary}
\newtheorem{proposition}[theorem]{Proposition}
\newtheorem{lemma}[theorem]{Lemma}
\newtheorem{example}[theorem]{Example}
\newtheorem{definition}[theorem]{Definition}
\theoremstyle{remark}
\newtheorem{remark}[theorem]{Remark}

\newcommand{\R}{\mathbb{R}}                 % Real numbers
\newcommand{\C}{\mathbb{C}}                   % Complex numbers
\newcommand{\N}{\mathbb{N}}                 % Natural numbers
\newcommand{\Z}{\mathbb{Z}}                 % Integers
\newcommand{\im}{\mathrm i}

\newcommand{\cS}{\mathcal{S}}

\newcommand{\text}[1]{\mathrm{#1}}
\begin{document}
%%%% --------------------------------------- Tittle and abstract ------------------------------------------%%%%
%%%%
\title{Periodic quantum graphs from the Bethe--Sommerfeld perspective}
\author{Pavel Exner$^{1,2}$, Ond\v{r}ej Turek$^{2,3,4}$}
\address{$^1$ Doppler Institute for Mathematical Physics and Applied Mathematics, Czech
Technical University in Prague, B\v{r}ehov\'{a} 7, 11519 Prague,
Czech Republic}
\address{$^2$ Department of Theoretical Physics, Nuclear Physics Institute CAS, 25068 \v{R}e\v{z}
near Prague, Czech Republic}
\address{$^3$ Bogoliubov Laboratory of Theoretical Physics, Joint Institute for Nuclear Research,
141980 Dubna, Russian Federation}
\address{$^4$ Laboratory for Unified Quantum Devices, Kochi University of Technology, 782-8502 Kochi, Japan}
\ead{exner@ujf.cas.cz,
o.turek@ujf.cas.cz}

\begin{abstract}
The paper is concerned with the number of open gaps in spectra of periodic quantum graphs. The well-known conjecture by Bethe and Sommerfeld (1933) says that the number of open spectral gaps for a system periodic in more than one direction is finite. To the date its validity is established for numerous systems, however, it is known that quantum graphs do not comply with this law as their spectra have typically infinitely many gaps, or no gaps at all. These facts gave rise to the question about the existence of quantum graphs with the `Bethe--Sommerfeld property', that is, featuring a nonzero finite number of gaps in the spectrum. In this paper we prove that the said property is impossible for graphs with the vertex couplings which are either scale-invariant or associated to scale-invariant ones in a particular way. On the other hand, we demonstrate that quantum graphs with a finite number of open gaps do indeed exist. We illustrate this phenomenon on an example of a rectangular lattice with a $\delta$ coupling at the vertices and a suitable irrational ratio of the edges. Our result allows to find explicitly a quantum graph with any prescribed exact number of gaps, which is the first such example to the date.
\end{abstract}


\vspace{2pc} \noindent{\it Keywords}: quantum graphs, periodic structure, Bethe--Sommerfeld conjecture, vertex coupling, Diophantine approximation

\submitto{\JPA}

%\maketitle

%%%%%%%%%%%%%%%%%%%%%%%%%%%%%%%%%%%%%%%%%%%%%%%%%%%%%%%%%%%%%%%%%%%%%%%%%%%%%%%%%%%%%%%%%%%%%%%%%%%%%%%%%%%%%%%
%%%%------------------------------------------ Introduction -----------------------------------------------%%%%
%%%%%%%%%%%%%%%%%%%%%%%%%%%%%%%%%%%%%%%%%%%%%%%%%%%%%%%%%%%%%%%%%%%%%%%%%%%%%%%%%%%%%%%%%%%%%%%%%%%%%%%%%%%%%%%
%  Section
%
\section{Introduction}\label{s:intro}

Quantum graphs are one of the fast developing areas of quantum physics,
the interest to them being driven both by their `practical' use in
modeling nanostructures and other physical objects, as well as by
theoretical reasons. They allow us to understand better various
quantum effects by analyzing them in the situation where the
configuration space has nontrivial geometrical and topological
properties. The literature concerning quantum graphs is extensive
and we limit ourselves to referring the reader to the recent
monograph \cite{BK13} as a guide to further ilumination.

While the quantum graph Hamiltonians describing particles `living'
on a metric graph share many properties with the `usual'
Schr\"odinger operators, this analogy is far from being complete; a
well-known example is the failure of the unique continuation
property \cite[Sec.~3.4]{BK13} that makes possible, for instance,
the existence of compactly supported eigenfunctions on infinite
graphs. This concerns, in particular, \emph{infinite periodic
graphs} the spectrum of which may not be purely absolutely
continuous containing flat bands, or infinitely degenerate
eigenvalues, and it is even possible that the absolutely continuous
part is empty as is the case for magnetic chain graphs with a
half-of-the-quantum flux through each chain element
\cite[Thm.~2.3]{EM15}.

Our goal in this paper is to investigate Hamiltonians of infinite periodic graphs from another point of view, namely the number of open gaps in their spectra. To begin with we recall the \emph{Bethe--Sommerfeld conjecture} \cite{BS33} put forward in the early days of the quantum theory, according to which a quantum system periodic in more than one direction --- with a slight abuse of terminology one usually speaks of  $\Z^\nu$-periodicity with $\nu\ge 2$ --- has a finite number of open gaps in the spectrum only. The reasoning behind the conjecture is based on the behavior of the spectral bands identified with the ranges of the dispersion curves or surfaces. Those at most touch for $\Z$-periodic systems while in higher dimensions they typically overlap making opening of gaps more and more difficult as we proceed to higher energies. This looked convincing and the property was taken for granted, although mathematically it proved to be a rather hard problem and it took decades before an affirmative answer was obtained for most cases of the `ordinary' Schr\"odinger operators --- see, for instance, \cite{DT82, HM98, Pa08, Sk79, Sk85} and references therein.

Discussing this question in the context of quantum graphs, the authors of \cite{BK13} recall the above mentioned heuristic argument (Sec.~4.7), however, they add immediately that this is not a `strict law'; in Sec.~5.1 of \cite{BK13} they illustrate this claim by examples of periodic graphs with an infinite number of resonant gaps created by a graph `decoration', the effect noticed first in the context of discrete graphs \cite{SA00} and later verified also for metric graphs\footnote{Cf. \cite{Ku05} and \cite[Sec.~5.1]{BK13} for more details. Let us add that decorations can produce an infinite number of spectral gaps also in systems that are $\Z$-periodic only \cite{AEL94}.}. In other words, we have examples of numerous situations in which the claim represented by the BS conjecture is false. The question thus arise whether it is a `law' at all, that is, whether there \emph{are} infinite periodic graphs having a \emph{finite nonzero} number of open gaps above the threshold of the spectrum. This is the topic we are going to discuss in the present paper; for the brevity of expression we will speak of those graphs as of graphs belonging to the \emph{Bethe--Sommerfeld class}, or simply \emph{Bethe--Sommerfeld graphs}.

We have two main conclusions. The first one concerns the
fact that the said property is sensitive to the type of vertex
coupling. Recall that the general self-adjoint coupling condition, commonly written as $(U-I)\Psi +\im(U+I)\Psi'=0$, can be decomposed into the Dirichlet, Neumann, and Robin parts \cite[Thm.~1.4.4]{BK13}; if the latter is absent we call such a coupling \emph{scale-invariant}.
 % -------------- %
\begin{theorem} \label{thm:absence}
An infinite periodic quantum graph does not belong to the
Bethe--Sommerfeld class if all the couplings at its vertices are
scale-invariant.
\end{theorem}
 % -------------- %
\noindent In fact, one can make a stronger claim. Given a graph
with general couplings we consider the same graph with the couplings
made scale-invariant by removing the Robin component in the
way described in Section~\ref{ss:general}.
Comparing their spectra, we find:

\begin{theorem}\label{thm:associated}
If an infinite periodic quantum graph with
scale-invariant couplings at the vertices has at least one open gap,
then adding a Robin component to the couplings
cannot produce a Bethe--Sommerfeld graph.
\end{theorem}

On the other hand, we are going to demonstrate that the said class is nonempty. Our second main result in this paper is expressed in the following claim.
 % -------------- %
\begin{theorem} \label{thm:existence}
Bethe--Sommerfeld graphs exist.
\end{theorem}
 % -------------- %
\noindent As it is usually the case with existence claims it is sufficient to present an example. With this aim we revisit in the second part of the paper the model introduced in \cite{Ex95} and further discussed in \cite{Ex96, EG96} describing a periodic lattice whose basic cell is a rectangle of the side ratio $\theta$ and the coupling in the vertices is of the $\delta$-type with a coupling constant $\alpha\in\R$. It was shown in the mentioned papers that the spectral properties of such a quantum graph depend on the number-theoretical properties of the ratio $\theta$. Here we are
going to demonstrate that if $\theta$ is badly approximable by rationals, there are values of $\alpha$ for which this graph belongs to the Bethe--Sommerfeld class. More than that, our construction makes it possible to find values of $\alpha$ for which the lattice graph in question has any prescribed number of gaps.

Before closing the introduction, let us recall that there are examples of the `usual' Schr\"odinger operators where the question about validity of the conjecture remains open, a prominent example being Laplacian in a periodically curved tube or a Schr\"odinger operator in a straight tube with a $\Z$-periodic potential. These systems are sometimes said to have a `mixed dimensionality' even if they are obviously periodic in one direction only, however they have a `two-dimensional' feature, namely that in the absence of potential or the deformation they have intersecting dispersion curves, which could suggest a BS-type behaviour. An analogue of such systems in the present context are $\Z$-periodic graphs with period cells connected by more than a single link for which the question about the Bethe-Sommerfeld property remains also open; note that Theorems~\ref{thm:absence} and \ref{thm:associated} apply to such graphs.

%%%%%%%%%%%%%%%%%%%%%%%%%%%%%%%%%%%%%%%%%%%%%%%%%%%%%%%%%%%%%%%%%%%%%%%%%%%%%%%%%%%%%%%%%%%%%%%%%%%%%%%%%%%%%%%
%%%%-------------------------------------- Absence of BS property -----------------------------------------%%%%
%%%%%%%%%%%%%%%%%%%%%%%%%%%%%%%%%%%%%%%%%%%%%%%%%%%%%%%%%%%%%%%%%%%%%%%%%%%%%%%%%%%%%%%%%%%%%%%%%%%%%%%%%%%%%%%
%  Section
%
\section{Absence of the Bethe--Sommerfeld property}\label{s:scaleinv}
\setcounter{equation}{0}

In this section we are going to prove Theorems~\ref{thm:absence} and \ref{thm:associated}.

\subsection{The ST-form of the coupling}

As it is common in the quantum graph theory the Hamiltonians we consider act as the (negative) second derivative on the graph edges with the domain consisting of functions which belong locally to the second Sobolev space and satisfy suitable coupling conditions at the vertices. For the purposes of the argument it is useful to write the vertex conditions, instead of the commonly used way mentioned in the introduction, in the so-called \emph{ST-form} proposed in~\cite{CET10}. Given a vertex of degree $n$, the vectors $\Psi$ and $\Psi'$ in $\C^n$ will again stand for the boundary values in the vertex,
 % -------------- %
$$
\Psi:=\left(\begin{array}{c}
\psi_1(0)\\
\vdots\\
\psi_n(0)
\end{array}\right)\,,\quad
\Psi':=\left(\begin{array}{c}
\psi_1'(0)\\
\vdots\\
\psi_n'(0)
\end{array}\right)\,,
$$
 % -------------- %
where the limits of the first derivatives are conventionally taken
in the outward direction. The coupling condition at the vertex can
be then written as
 % -------------- %
\begin{equation}\label{ST}
\left(\begin{array}{cc}
I^{(r)} & T \\
0 & 0
\end{array}\right)
\Psi'= \left(\begin{array}{cc}
S & 0 \\
-T^* & I^{(n-r)}
\end{array}\right)
\Psi
\end{equation}
 % -------------- %
for certain $r$, $S$, and $T$, where the symbol $I^{(r)}$ denotes the identity matrix of order $r$ and the matrix $S$ is Hermitian. The condition~(\ref{ST}) allows us to single out scale-invariant couplings; it is easy to see that the coupling has this property if and only if $S=0\:$~\cite{CET10b}. In particular, the on-shell scattering matrix $\cS(k)$ for the vertex in question is in the $ST$-formalism given by
 % -------------- %
\begin{equation}\label{S(S,T,k)}
\cS(k)= -I^{(n)} +2\left(\begin{array}{c} I^{(r)} \\ T^*
\end{array}\right) \left(I^{(r)}+TT^*-\frac{1}{\im k}S\right)^{-1}
\left(\begin{array}{cc} I^{(r)} & T\end{array}\right)
\end{equation}
 % -------------- %
and it is obvious that $\cS(k)$ is independent of $k$ \emph{iff} $S=0$.

The spectrum is obtained using the Bloch-Floquet theory. The way to do that is well known, cf.~\cite[Sec.~4.2]{BK13} and references therein, we describe it nevertheless briefly here to make the paper self-contained. We assume that the graph is locally finite and consider its elementary cell; cutting it out from the original periodic graph we get a finite family of pairs of `antipodal' vertices related mutually by the action of the corresponding translation group. Each such pair of vertices $(v_-,v_+)$ can be regarded as a single vertex with the boundary
conditions
 % -------------- %
\begin{equation}\label{Floquet}
\psi(v_+)=\e^{\im\vartheta_l}\psi(v_-)\,,\quad
\psi'(v_+)=\e^{\im\vartheta_l}\psi'(v_-)
\end{equation}
 % -------------- %
for some $\vartheta_l\in(-\pi,\pi]$, where $l=1,\ldots,\nu$, and $\nu$ is the dimension of translation group associated with graph periodicity. The pair of edges with the endpoints $v_\pm$ can be turned into a single edge by identifying these endpoints, and the acquired phase $\vartheta_l$ coming from the conditions~(\ref{Floquet}) can be also regarded as being induced by a magnetic potential. Let us denote the graph obtained in this way from the elementary cell by $\Gamma$, and the number of its edges by $E$. Regarding each edge of $\Gamma$ as a pair of two directed edges (bonds) of opposite orientations and indexing the bonds by $1,2,\ldots,2E$, we consider the $2E\times2E$ matrices $\mathbf{L}$, $\mathbf{\Theta}$ and $\mathbf{S}(k)$ which are defined in the following way. The matrix $\mathbf{L}$ is a diagonal matrix whose $j$-th diagonal entry is the length of the $j$-th bond. The diagonal matrix $\mathbf{\Theta}$ has the entries $\vartheta_l$ and $-\vartheta_l$ at the pair of the diagonal positions corresponding to the $l$-th edge of $\Gamma$ created by the mentioned vertex identification; all its other entries are zero. Finally, the matrix $\mathbf{S}(k)$ is the bond scattering matrix, which contains directed edge-to-edge scattering coefficients. In this way, each element of the matrix $\mathbf{S}(k)$ corresponds to a certain entry of the scattering matrix at a certain vertex of the elementary cell, cf.~\cite[eq.~(2.1.15)]{BK13}. Recall that the bond scattering matrix $\mathbf{S}(k)$ is unitary.
Having introduced the matrices $\mathbf{L}$, $\mathbf{\Theta}$, and $\mathbf{S}(k)$, we define the function $F(k;\vec{\vartheta})$ as
 % -------------- %
\begin{equation}\label{F}
F(k;\vec{\vartheta}):=\det\left(\mathbf{I}-\e^{\im(\mathbf{\Theta}+k\mathbf{L})}\mathbf{S}(k)\right)\,;
\end{equation}
 % -------------- %
this allows us to write the spectral condition in the form
 % -------------- %
\begin{equation}\label{condition}
k^2\in\sigma(H) \quad\Leftrightarrow\quad
(\exists\vec{\vartheta}\in(-\pi,\pi]^\nu)(F(k;\vec{\vartheta})=0)\,.
\end{equation}
 % -------------- %
Note that the function $F(k;\vec{\vartheta})$ is in general complex, however, one can consider a real-valued function instead, dividing $F(k;\vec{\vartheta})$ by $\sqrt{\det(\e^{\im(\mathbf{\Theta}+k\mathbf{L})}\mathbf{S}(k))}$, cf.~\cite[Rem.~2.1.10]{BK13}.

\subsection{Graphs with scale-invariant couplings}

As we have indicated, the described way of treating periodic quantum graphs is pretty standard. It has an important advantage, namely that it allows one to analyze properties of the ergodic flow on the torus associated with such a system. This idea can be traced back to Barra and Gaspard \cite{BG00} and it was recently used by Band and Berkolaiko \cite{BB13} to derive a deep result about spectral universality for periodic quantum graphs with the simplest vertex coupling, usually referred to as Kirchhoff or Neumann. We are going to use an argument analogous to that of \cite{BB13} in Proposition~\ref{Prop.FT} below.

Consider first the case of a periodic graph with scale-invariant couplings at all the vertices. The scale-invariance assumption implies
that the scattering matrix at each graph vertex is independent of $k$ and the same is naturally true for the matrix $\mathbf{S}(k)$ entering
formula~(\ref{F}). From now on, let $F_0(k;\vec{\vartheta})$ denote the left hand side of formula~(\ref{F}) for a graph with scale-invariant couplings. The function value $F_0(k;\vec{\vartheta})$
thus depends on the vectors $\vec{\vartheta}$ and $(k\ell_0,k\ell_1,\ldots,k\ell_d)$, where $\{\ell_0,\ell_1,\ldots,\ell_d\},\: d+1\le E$, is the set of mutually different edge lengths of $\Gamma$. Moreover, the value $F(k;\vec{\vartheta})$ is $2\pi$-periodic in each of the terms $k\ell_0,k\ell_1,\ldots,k\ell_d$.
Let us define
\begin{equation}\label{phi}
\vec{\phi}(k)=(\{k\ell_0\}_{(2\pi)},\{k\ell_1\}_{(2\pi)},\ldots,\{k\ell_d\}_{(2\pi)}),
\end{equation}
where the symbol $\{x\}_{(2\pi)}$ stands for the difference between $x$ and the nearest integer multiple of $2\pi$, i.e.
 % -------------- %
\begin{equation}\label{fract.part}
\{x\}_{(2\pi)}=x-2\pi m \quad \text{if}\;\;
x\in((2m-1)\pi,(2m+1)\pi]\,.
\end{equation}
 % -------------- %
Since the value $F_0(k;\vec{\vartheta})$ depends on the vectors $\vec{\phi}(k)$
and $\vec{\vartheta}$ only, it is convenient to introduce a function
\begin{equation}\label{Phi}
\Phi_0(\vec{\phi}(k);\vec{\vartheta})=F_0(k;\vec{\vartheta})
\end{equation}
and write the spectral condition~(\ref{condition}) in the form
 % -------------- %
\begin{equation}\label{SC Phi}
k^2\in\sigma(H_0) \quad\Leftrightarrow\quad
\left(\exists\vec{\vartheta}\in(-\pi,\pi]^\nu\right)\left(\Phi_0(\vec{\phi}(k);\vec{\vartheta})=0\right).
\end{equation}
 % -------------- %

\begin{lemma}\label{Lem. phi}
Let $\vec{\phi}(k)$ be given by (\ref{phi}). For any $k>0$, $C>0$ and $\delta>0$ there is a $k'>C$ such that $\|\vec{\phi}(k')-\vec{\phi}(k)\|_\infty<\delta$.
\end{lemma}

\begin{proof}
We shall prove that
 % -------------- %
\begin{equation}\label{implication}
\hspace{-6em} (\forall k>0)(\forall C>0)(\forall\delta>0)(\exists k'>C)(\forall
j\in\{0,1,\ldots,d\})\left(\big|\left\{k'\ell_j-k\ell_j\right\}_{(2\pi)}\big|<\delta\right)\!,
\end{equation}
 % -------------- %
where the symbol $\{\cdot\}_{(2\pi)}$ was defined in~(\ref{fract.part}).
We will use the simultaneous version of the Dirichlet's approximation theorem. First of all, we set
 % -------------- %
\begin{equation}\label{alpha}
\alpha_j=\frac{\ell_j}{\ell_0}
\end{equation}
 % -------------- %
for all $j=1,\ldots,d$. The said theorem guarantees for any $\alpha_1,\ldots,\alpha_d\in\R$ and for any natural number $N$ the existence of integers $p_1,\ldots,p_d,\:q\in\Z$, $\,1\leq q\leq N$, such that
 % -------------- %
\begin{equation}\label{Dirichlet}
\left|\alpha_j-\frac{p_j}{q}\right|\leq\frac{1}{qN^{1/d}}\,.
\end{equation}
 % -------------- %
Let $k$, $C$ and $\delta$ be given and choose $m$ as an integer with the property that
 % -------------- %
\begin{equation}\label{m}
m>\frac{\ell_0C}{2\pi}\,.
\end{equation}
 % -------------- %
Once $m$ is fixed, the number $N$ can be taken as any integer satisfying
 % -------------- %
\begin{equation}\label{N}
N>\left(\frac{2\pi}{\delta}m\right)^d\,.
\end{equation}
 % -------------- %
Let $q$ be the integer from the simultaneous version of the Dirichlet's approximation theorem corresponding to $N$ chosen according to~(\ref{N}). Notice that $q$ depends on $N$, and therefore also on $\delta$. For this $q$, we define $k'_\delta$ as follows,
 % -------------- %
$$
k'_\delta:=k+2\pi m\frac{q}{\ell_0}\,.
$$
 % -------------- %
Our aim is to show that $k'_\delta$ satisfies the following two conditions:
 % -------------- %
\begin{eqnarray}
\phantom{AAAAAAAAii}k'_\delta>C \label{(1)} && \\
\big|\left\{k'_\delta\ell_j-k\ell_j\right\}_{(2\pi)}\big|<\delta\,, &&
 \quad j=0,1,\ldots,d \label{(2)}
\end{eqnarray}
 % -------------- %
Applying the definition of $k'_\delta$, the inequality $q\geq1$ and the assumption~(\ref{m}), we get
 % -------------- %
$$
k'_\delta=k+2\pi m\frac{q}{\ell_0}>2\pi m\frac{q}{\ell_0}\geq2\pi m\frac{1}{\ell_0}>C\,,
$$
 % -------------- %
in other words, condition~(\ref{(1)}) holds true. Let us proceed to condition~(\ref{(2)}). We have
 % -------------- %
$$
k'_\delta\ell_j-k\ell_j=2\pi m\frac{q}{\ell_0}\ell_j=2\pi mq\alpha_j\,,
$$
 % -------------- %
where $\alpha_j$ was introduced in equation~(\ref{alpha}). Since $\left|\{x\}_{(2\pi)}\right|\leq|x-2\pi p|$ holds obviously for all $x\in\R$ and $p\in\mathbb{Z}$, we obtain in particular
 % -------------- %
$$
\left|\left\{k'_\delta\ell_j-k\ell_j\right\}_{(2\pi)}\right|
\leq|k'_\delta\ell_j-k\ell_j-2\pi mp_j|
=|2\pi mq\alpha_j-2\pi mp_j|
=2\pi mq\left|\alpha_j-\frac{p_j}{q}\right|
$$
 % -------------- %
for $p_1,\ldots,p_d$ denoting the integers from (\ref{Dirichlet}).
Consequently, the inequality~(\ref{Dirichlet}) and the assumption~(\ref{N}) imply
 % -------------- %
\begin{equation}\label{delta}
\left|\left\{k'_\delta\ell_j-k\ell_j\right\}_{(2\pi)}\right|\leq2\pi
mq\frac{1}{qN^{1/d}}=\frac{2\pi m}{N^{1/d}}<\delta\,,
\end{equation}
 % -------------- %
which proves condition~(\ref{(2)}). The claim~(\ref{implication}) thus holds true.
\end{proof}

 % -------------- %
\begin{proposition}\label{Prop.FT}
Let $H_0$ be a Hamiltonian of a periodic quantum graph with scale-invariant couplings at all the vertices. Then the following holds:
\begin{itemize}
\item[(i)] If $\sigma(H_0)$ contains a gap, then it contains infinitely many gaps.
\item[(ii)] Furthermore, if $\sigma(H_0)$ has a gap of size $s$ in terms of momentum, then for every $\epsilon>0$ there is an infinite sequence of gaps of sizes at least $s-\epsilon$ (in terms of momentum).
\item[(iii)] In particular, if all the graph edge lengths are rationally dependent, then the momentum spectrum is periodic.
\end{itemize}
\end{proposition}
 % -------------- %
\begin{proof}
The claim (i) is obviously a straighforward consequence of (ii). To prove (ii), let us assume that $\sigma(H_0)$ has a gap of size $s$ in terms of $k$, i.e., there is a $k_0>0$ such that $k^2\notin\sigma(H_0)$ for all $k\in(k_0-\frac{s}{2},k_0+\frac{s}{2})$. According to~(\ref{SC Phi}), the gap condition reads
 % -------------- %
$$
k^2\notin\sigma(H_0) \quad\Leftrightarrow\quad
(\forall\vec{\vartheta}\in(-\pi,\pi]^\nu)(|\Phi_0(\vec{\phi}(k);\vec{\vartheta})|>0)\,.
$$
 % -------------- %
Since $|\Phi_0(\cdot;\cdot)|$ is a continuous function, it attains minimum on any compact interval. In particular, for any given $\epsilon>0$ there exists the minimum
 % -------------- %
$$
\min_{k\in\left[k_0-\frac{s-\epsilon}{2},k_0+\frac{s-\epsilon}{2}\right],\;\vec{\vartheta}\in[-\pi,\pi]^\nu}|\Phi_0(\vec{\phi}(k);\vec{\vartheta})|=\gamma\,.
$$
 % -------------- %
The value $\gamma$ is positive, because all values in the interval $\bigl[k_0-\frac{s-\epsilon}{2},k_0+\frac{s-\epsilon}{2}\bigr]$ correspond to a gap.
Moreover, the Brillouin zone has the structure of a torus, hence the function $|\Phi_0(\vec{\phi}(k);\cdot)|$ is periodic with the period $2\pi$ in every component of the vector $\vec{\vartheta}$, which in particular means that the same value of minimum is attained also at the left-open interval $(-\pi,\pi]^\nu$. Hence we obtain
 % -------------- %
\begin{equation}\label{gap k FT}
\hspace{-4em} \left(\forall x\in\left[-\frac{s-\epsilon}{2},\frac{s-\epsilon}{2}\right]\right)\left(\forall\vec{\vartheta}\in(-\pi,\pi]^\nu\right)\left(|\Phi_0(\vec{\phi}(k_0+x);\vec{\vartheta})|\geq\gamma>0\right)\,.
\end{equation}
 % -------------- %
Now we use Lemma~\ref{Lem. phi}, which guarantees that for every $C>0$ one can find a $k'>C$ such that for all $j\in\{0,1,\ldots,d\}$, the quantity $\big|\left\{k'\ell_j-k\ell_j\right\}_{(2\pi)}\big|$ is as small as required.
With regard to a trivial identity $\left\{(k'+x)\ell_j-(k+x)\ell_j\right\}_{(2\pi)}=\left\{k'\ell_j-k\ell_j\right\}_{(2\pi)}$ for $x\in\R$, also the quantity $\bigl|\left\{(k'+x)\ell_j-(k+x)\ell_j\right\}_{(2\pi)}\bigr|$ can be as small as required for all $x$.
This fact together with the continuity of $\Phi_0$ implies that for every given $\gamma>0$, one can find a $k'>C$ with the property
\begin{equation}\label{gap k' k delta FT}
\max_{x\in\bigl[-\frac{s-\epsilon}{2},\frac{s-\epsilon}{2}\bigr],\;\vec{\vartheta}\in(-\pi,\pi]^\nu}\left|\Phi_0(\vec{\phi}(k'+x);\vec{\vartheta})-\Phi_0(\vec{\phi}(k_0+x);\vec{\vartheta})\right|<\frac{\gamma}{2}\,.
\end{equation}
Now we apply the triangle inequality together with (\ref{gap k FT}) and (\ref{gap k' k delta FT}) to obtain the estimate
\begin{eqnarray*}
\hspace{-4em} \left|\Phi_0(\vec{\phi}(k'+x);\vec{\vartheta})\right|&\geq\left|\Phi_0(\vec{\phi}(k_0+x);\vec{\vartheta})\right|-\left|\Phi_0(\vec{\phi}(k'+x);\vec{\vartheta})-\Phi_0(\vec{\phi}(k_0+x);\vec{\vartheta})\right| \\
&>\gamma-\frac{\gamma}{2}=\frac{\gamma}{2}>0
\end{eqnarray*}
 % -------------- %
for all $x\in[-\frac{s-\epsilon}{2},\frac{s-\epsilon}{2}]$ and $\vec{\vartheta}\in(-\pi,\pi]^\nu$.
To sum up, for any $C>0$ one can find a
$k'>C$ such that $k^2\notin\sigma(H)$ for all $k\in[k'-\frac{s-\epsilon}{2},k'+\frac{s-\epsilon}{2}]$.
This proves the existence of infinitely many gaps of sizes at least $s-\epsilon$ in terms of $k$, given the fact that the
operator in question is unbounded and the resolvent set of $H$ is open.

It remains to prove (iii). If all the lengths are rationally dependent, there exists an elementary length $L>0$ and integers $m_j\in\N$ such that $\ell_j=m_jL$ holds for $j=0,1,\ldots,d$. Hence $\left(k+\frac{2\pi}{L}\right)\ell_j=k\ell_j+2\pi m_j$ which implies
 % -------------- %
$$
\left\{\left(k+\frac{2\pi}{L}\right)\ell_j\right\}_{(2\pi)}=\{k\ell_j\}_{(2\pi)}
$$
 % -------------- %
for all $j=0,1,\ldots,d$. This means that
$\vec{\phi}(k)$ is periodic with period $2\pi/L$, and consequently, the spectrum has a periodic structure in terms of the momentum.
\end{proof}
 % -------------- %
\begin{corollary}
Theorem~\ref{thm:absence} is valid.
\end{corollary}
 % -------------- %

\subsection{The case of general vertex couplings} \label{ss:general}

Our next aim is to show that the Bethe--Sommerfeld property can be excluded also for graphs with vertex couplings from a wider class. We begin with the following definition.
 % -------------- %
\begin{definition}
Let a vertex coupling be given by condition~(\ref{ST}). The \emph{associated scale-invariant vertex coupling} is given by condition
 % -------------- %
\begin{equation}\label{scale invariant}
\left(\begin{array}{cc}
I^{(r)} & T \\
0 & 0
\end{array}\right)
\Psi'= \left(\begin{array}{cc}
0 & 0 \\
-T^* & I^{(n-r)}
\end{array}\right)
\Psi\,.
\end{equation}
 % -------------- %
\end{definition}
 % -------------- %
\noindent In other words, the coupling associated to a given~(\ref{ST}) is obtained by removing the Robin part represented by the square matrix $S$.

In the following proposition we show that the scattering matrix referring to~(\ref{ST}) decomposes into a constant part and a part that vanishes as
$k\to\infty$. This observation is useful for dealing with high momenta values, $k\gg1$, note that this is the regime crucial from the viewpoint of the Bethe--Sommerfeld property.
 % -------------- %
\begin{proposition}\label{Prop. S-matrix}
Consider a quantum graph vertex with a general coupling described by the condition~(\ref{ST}). Its scattering matrix satisfies
 % -------------- %
\begin{equation}\label{S-matrix 1/k}
\cS(k)=\cS_0+\frac{1}{k}\cS_1(k)\,,
\end{equation}
 % -------------- %
where
 % -------------- %
\begin{equation}\label{S0}
\cS_0=-I^{(n)}+ 2\left(\begin{array}{c} I^{(r)} \\
T^* \end{array}\right) \left(I^{(r)}+TT^*\right)^{-1}
\left(\begin{array}{cc} I^{(r)} & T\end{array}\right)
\end{equation}
 % -------------- %
is the constant scattering matrix corresponding to the associated scale-invariant vertex coupling~(\ref{scale invariant}), and
 % -------------- %
\begin{equation}\label{S1k}
\hspace{-4em} \cS_1(k)=-2\im\left(\begin{array}{c} I^{(r)} \\ T^* \end{array}\right)
\left(I^{(r)}+TT^*\right)^{-1}S\left(I^{(r)}+TT^*-\frac{1}{\im
k}S\right)^{-1} \left(\begin{array}{cc} I^{(r)} &
T\end{array}\right)\,.
\end{equation}
 % -------------- %
Moreover, the matrix function $k\mapsto\cS_1(k)$ is bounded on the interval $[1,\infty)$.
\end{proposition}
 % -------------- %
\begin{proof}
The scattering matrix $\cS(k)$ obeys equation~(\ref{S(S,T,k)}). Using the identity
$$
\left(I^{(r)}+TT^*-\frac{1}{\im k}S\right)^{-1}
=\left(I^{(r)}+TT^*\right)^{-1}+\frac{1}{\im k}(I^{(r)}+TT^*)^{-1}S\left(I^{(r)}-\frac{1}{\im k}(I^{(r)}+TT^*)^{-1}S\right)^{-1}\,,
$$
we obtain~(\ref{S-matrix 1/k}) for $\cS_0$ and $\cS_1(k)$ given by (\ref{S0}) and (\ref{S1k}), respectively.
Finally, the boundedness of $\cS_1(k)$ on $[1,\infty)$ is a straightforward consequence of the continuity of $k\mapsto\cS_1(k)$ and the existence of the limit
 % -------------- %
$$
\lim_{k\to\infty}\cS_1(k)=-2\im\left(\begin{array}{c} I^{(r)} \\ T^*
\end{array}\right)
\left(I^{(r)}+TT^*\right)^{-1}S\left(I^{(r)}+TT^*\right)^{-1}
\left(\begin{array}{cc} I^{(r)} & T\end{array}\right)\,.
$$
 % -------------- %
\end{proof}

Each entry of the matrix $\mathbf{S}(k)$, appearing in~(\ref{F}), is equal by definition to a certain entry of $\cS(k)$ for some vertex of $\Gamma$. Therefore, Proposition~\ref{Prop. S-matrix} allows us to decompose the matrix $\mathbf{S}(k)$ in a way similar to~(\ref{S-matrix 1/k}),
 % -------------- %
\begin{equation}\label{SS}
\mathbf{S}(k)=\mathbf{S}_0+\frac{1}{k}\mathbf{S}_1(k)\,,
\end{equation}
 % -------------- %
where $\mathbf{S}_0$ is a constant unitary matrix, corresponding to the same graph with the associated scale-invariant couplings at its vertices,
and $\mathbf{S}_1(k)$ is a matrix that is bounded on $[1,\infty)$ as a function of $k$.

\begin{proposition}
The quantity $F(k;\vec{\vartheta})$ of (\ref{F}) can be expressed as
 % -------------- %
\begin{equation}\label{F 1/k}
F(k;\vec{\vartheta})=\Phi_0(\vec{\phi}(k);\vec{\vartheta})+\frac{1}{k}F_1(k;\vec{\vartheta})\,,
\end{equation}
 % -------------- %
where
 % -------------- %
\begin{equation*}
\Phi_0(\vec{\phi}(k);\vec{\vartheta}):=\det\left(\mathbf{I}-\e^{\im(\mathbf{\Theta}+k\mathbf{L})}\mathbf{S}_0\right)
\end{equation*}
 % -------------- %
and the function $k\mapsto F_1(k;\vec{\vartheta})$ is continuous and bounded on $[1,\infty)$.
\end{proposition}
 % -------------- %
\begin{proof}
According to (\ref{SS}) we have
 % -------------- %
$$
F(k;\vec{\vartheta})=\det\left(\mathbf{I}-\e^{\im(\mathbf{\Theta}+k\mathbf{L})}\mathbf{S}(k)\right)=\det\left(M_0+\frac{1}{k}M_1\right)
$$
 % -------------- %
with $M_0:=\mathbf{I}-\e^{\im(\mathbf{\Theta}+k\mathbf{L})}\mathbf{S}_0$ and $M_1:=-\e^{\im(\mathbf{\Theta}+k\mathbf{L})}\mathbf{S}_1(k)$.
The expansion of $\det\left(M_0+\frac{1}{k}M_1\right)$ takes the form $\det(M_0)+\frac{1}{k}F_1(k;\vec{\vartheta})$, where $F_1(k;\vec{\vartheta})$ is a sum of products of entries of the matrices $M_0$ and $M_1$ multiplied by non-negative powers of $1/k$. Since the entries of $M_0$ and $M_1$ are continuous and bounded on $[1,\infty)$ (cf.~(\ref{SS})), the function $F_1(k;\vec{\vartheta})$ has the same property.
\end{proof}

 % -------------- %
\begin{theorem}\label{thm:gaps general}
Consider a periodic graph with general couplings at the vertices and denote its spectrum as $\sigma(H)$. Let further $\sigma(H_0)$ be the spectrum of the same graph, in which all vertex couplings are replaced by the associated scale-invariant couplings. Then the following claims hold true:
 % -------------- %
\begin{itemize}
\item[(i)] If $\sigma(H_0)$ has an open gap, then $\sigma(H)$ has infinitely many gaps.
\item[(ii)] In particular, if $\sigma(H_0)$ has a gap of size $s$ in terms of momentum, then for every $\epsilon>0$ there are infinitely many gaps of $\sigma(H)$ of sizes at least $s-\epsilon$ (in terms of momentum).
\end{itemize}
 % -------------- %
\end{theorem}
 % -------------- %
\begin{proof}
Since (i) is a straightforward consequence of (ii), it suffices to prove (ii).
Let $H_0$ have a gap of size $s$ in terms of $k$, i.e., let there be a $k_0$ such that
$k^2\notin\sigma(H_0)$ for all $k\in\bigl(k_0-\frac{s}{2},k_0+\frac{s}{2}\bigr)$. In the same way as in the proof of
Proposition~\ref{Prop.FT}(ii), one can show that for any $\epsilon>0$, there is a $\gamma>0$ such that
 % -------------- %
\begin{equation}\label{gap k}
\hspace{-4em} \left(\forall x\in\left[-\frac{s-\epsilon}{2},\frac{s-\epsilon}{2}\right]\right)\left(\forall\vec{\vartheta}\in(-\pi,\pi]^\nu\right)\left(|\Phi_0(\vec{\phi}(k_0+x);\vec{\vartheta})|\geq\gamma>0\right)
\end{equation}
 % -------------- %
(cf.~(\ref{gap k FT})), and then demonstrate that for any $C>0$ there is a $k'>C$ such that
 % -------------- %
\begin{equation}\label{gap k' k delta}
\left|\Phi_0(\vec{\phi}(k'+x);\vec{\vartheta})-\Phi_0(\vec{\phi}(k_0+x);\vec{\vartheta})\right|<\frac{\gamma}{2}
\end{equation}
 % -------------- %
for all $x\in\bigl[-\frac{s-\epsilon}{2},\frac{s-\epsilon}{2}\bigr]$ and $\vec{\vartheta}\in(-\pi,\pi]^\nu$ (cf.~(\ref{gap k' k delta FT})). Let us limit ourselves to large values $C$, specifically, to the values $C$ with the property
 % -------------- %
\begin{equation}\label{k' gamma/4}
k'>C-\frac{s}{2} \quad\Rightarrow\quad
\frac{|F_1(k';\vec{\vartheta})|}{k'}<\frac{\gamma}{4}\,,
\end{equation}
 % -------------- %
where $F_1(k;\vec{\vartheta})$ is the term appearing in equation~(\ref{F 1/k}). Now we apply twice the triangle inequality to the decomposition (\ref{F 1/k}) and after that we use inequalities~(\ref{gap k}), (\ref{gap k' k delta}), and (\ref{k' gamma/4}). In this way we obtain
 % -------------- %
\begin{eqnarray*}
\hspace{-6em} |F(k'+x;\vec{\vartheta})|&\geq\left|\Phi_0(\vec{\phi}(k'+x);\vec{\vartheta})\right|-\frac{|F_1(k'+x;\vec{\vartheta})|}{k'+x} \\
&\geq\left|\Phi_0(\vec{\phi}(k_0+x);\vec{\vartheta})\right|-\left|\Phi_0(\vec{\phi}(k'+x);\vec{\vartheta})-\Phi_0(\vec{\phi}(k_0+x);\vec{\vartheta})\right|-\frac{|F_1(k'+x;\vec{\vartheta})|}{k'+x} \\
&>\gamma-\frac{\gamma}{2}-\frac{\gamma}{4}=\frac{\gamma}{4}>0
\end{eqnarray*}
 % -------------- %
for all $x\in\bigl[-\frac{s-\epsilon}{2},\frac{s-\epsilon}{2}\bigr]$ and $\vec{\vartheta}\in(-\pi,\pi]^\nu$.
That is, to any sufficiently large $C>0$ one can
find a $k'>C$ such that $k^2\notin\sigma(H)$ for all $k\in[k'-\frac{s-\epsilon}{2},k'+\frac{s-\epsilon}{2}]$, which proves the
result.
\end{proof}

\begin{corollary}
Theorem~\ref{thm:associated} is valid.
\end{corollary}

\medskip

Let us finally remark that the gaps of $\sigma(H)$ seem to asymptotically coincide with those of $\sigma(H_0)$, however, we are not going to pursue this question here.

%%%%%%%%%%%%%%%%%%%%%%%%%%%%%%%%%%%%%%%%%%%%%%%%%%%%%%%%%%%%%%%%%%%%%%%%%%%%%%%%%%%%%%%%%%%%%%%%%%%%%%%%%%%%%%%
%%%%--------------------------------------- Rectangular lattices ------------------------------------------%%%%
%%%%%%%%%%%%%%%%%%%%%%%%%%%%%%%%%%%%%%%%%%%%%%%%%%%%%%%%%%%%%%%%%%%%%%%%%%%%%%%%%%%%%%%%%%%%%%%%%%%%%%%%%%%%%%%
%  Section
%
\section{Number theoretic preliminaries}\label{s::prelim}
\setcounter{equation}{0}

Before turning to our second main topic---establishing the existence of Bethe--Sommerfeld graphs---we need to introduce some number-theoretic notions on which the subsequent spectral analysis will rely
substantially. A number $\theta\in\mathbb{R}$ is called \emph{badly approximable} if there exists a $c>0$ such that
 % -------------- %
$$
\left|\theta-\frac{p}{q}\right|>\frac{c}{q^2}
$$
 % -------------- %
for all $p,q\in\mathbb{Z}$ with $q\neq0$. An irrational number $\theta$ is badly approximable if and only if the elements of its continued-fraction representation $[c_0,c_1,c_2,c_3,\ldots]$ are bounded~\cite{Kh64}. With our goal in mind we note that the badly approximable numbers emerged in \cite{Ex96} as the only ratios $\frac{a}{b}$ for which the spectrum of a rectangular lattice with edges $a$ and $b$ and $\delta$ couplings in the vertices may have a finite number of gaps.
 % -------------- %

\medskip

The so-called \emph{Markov constant} $\mu(\theta)$ of $\theta\in\mathbb{R}$ is defined as
 % -------------- %
\begin{equation}\label{mu(theta)}
\mu(\theta)=\inf\left\{c>0\;\left|\;\left(\exists_\infty(p,q)\in\mathbb{Z}^2\right)\left(\left|\theta-\frac{p}{q}\right|<\frac{c}{q^2}\right)\right.\right\},
\end{equation}
 % -------------- %
with $\exists_\infty$ meaning ``there exist infinitely many''.
The Markov constant is sometimes denoted by $\nu(\theta)$, cf.~\cite{Ca57}. Notice that $\mu(\theta)>0$ if and only if $\theta$ is badly approximable. Since every $\theta\in\mathbb{Q}$ has trivially $\mu(\theta)=0$, some authors define $\mu(\theta)$ only for $\theta$ being irrational.

Recall that by a theorem of Hurwitz \cite{Hu1891} for every irrational number $\theta$ there are infinitely many $(p,q)\in\mathbb{Z}^2$ such that $\left|\theta-\frac{p}{q}\right|<\frac{1}{\sqrt{5}q^2}$, in other words, $\mu(\theta)\leq\frac{1}{\sqrt{5}}$ holds for any $\theta\in\mathbb{R}$.

We say that $\theta,\theta'\in\mathbb{R}$ are \emph{equivalent} if there are integers $r,s,t,u$ such that
 % -------------- %
\begin{equation}\label{Markov equivalence}
\theta=\frac{r\theta'+s}{t\theta'+u} \;\quad\text{and}\quad ru-ts=\pm1.
\end{equation}
 % -------------- %
According to \cite[Thm.~IV]{Ca57}, $\theta,\theta'\in(0,1)$ are equivalent if and only if their continued fractions take the form
 % -------------- %
\begin{equation}\label{equiv. cont. fr.}
\begin{array}{c}
\theta=[0;a_1,a_2,\ldots,a_l,c_1,c_2,\ldots] \\
\theta'=[0;b_1,b_2,\ldots,b_m,c_1,c_2,\ldots]
\end{array}
\end{equation}
 % -------------- %
for suitable $l,m$ and $a_1,\ldots,a_l$, $b_1,\ldots,b_m$, and $c_1,c_2,\ldots$. One can prove that if $\theta$ and $\theta'$ are equivalent, then $\mu(\theta)=\mu(\theta')$; cf.~\cite[p.~11]{Ca57}. The particular choice $r=u=0$ and $s=t=1$ in equation (\ref{Markov equivalence}) establishes the equivalence of the numbers $\theta$ and $\theta^{-1}$; hence
 % -------------- %
\begin{equation}\label{Markov rec.}
\mu(\theta)=\mu(\theta^{-1}).
\end{equation}
 % -------------- %
Now we will introduce a function $v:\R\to\R_+$ the values $\upsilon(\theta)$ of which will play an important role in the analysis of our spectral problem; they can be regarded as a one-sided version of the Markov constant.
 % -------------- %
\begin{definition}
For any $\theta>0$, we set
 % -------------- %
\begin{equation}\label{upsilon}
\upsilon(\theta):=\inf\left\{c>0\;\left|\;\left(\exists_\infty(p,q)\in\mathbb{Z}^2\right)\left(0<\theta-\frac{p}{q}<\frac{c}{q^2}\right)\right.\right\}.
\end{equation}
 % -------------- %
\end{definition}
 % -------------- %
\begin{proposition}
For every $\theta>0$, we have
 % -------------- %
\begin{eqnarray}
\upsilon(\theta)=\inf\left\{c>0\;\left|\;\left(\exists_\infty m\in\mathbb{N}\right)\left(m(m\theta-\lfloor m\theta\rfloor)<c\right)\right.\right\}, \label{upsilon i} \\
\upsilon(\theta^{-1})=\inf\left\{c>0\;\left|\;\left(\exists_\infty m\in\mathbb{N}\right)\left(m(\lceil m\theta\rceil-m\theta)<c\right)\right.\right\}, \label{upsilon ii} \\
\mu(\theta)=\min\{\upsilon(\theta),\upsilon(\theta^{-1})\}, \label{upsilon iii}
\end{eqnarray}
 % -------------- %
where $\lfloor\cdot\rfloor$ and $\lceil\cdot\rceil$ are the floor and the ceiling function, respectively.
\end{proposition}
 % -------------- %
\begin{proof}
One can see easily that the right-hand side of (\ref{upsilon}) will remain unchanged if we assume $q>0$ and $p$ is replaced with $\lfloor q\theta\rfloor$, i.e.,
 % -------------- %
$$
\upsilon(\theta)=\inf\left\{c>0\;\left|\;\left(\exists_\infty q\in\mathbb{N}\right)\left(\theta-\frac{\lfloor q\theta\rfloor}{q}<\frac{c}{q^2}\right)\right.\right\}.
$$
 % -------------- %
In this way we obtain formula (\ref{upsilon i}).

Let us next prove (\ref{upsilon ii}). It follows from the definition that the left-hand side of (\ref{upsilon ii}) equals
 % -------------- %
\begin{equation}\label{LHS}
\mathrm{LHS}=\inf\left\{c>0\;\left|\;\left(\exists_\infty(p,q)\in(\mathbb{Z}\backslash\{0\})^2\right)\left(q(q\theta^{-1}-p)<c\right)\right.\right\}.
\end{equation}
 % -------------- %
At the same time, in analogy with the previous step, it is easy to see that the right-hand side of (\ref{upsilon ii}) is equal to
 % -------------- %
\begin{equation}\label{RHS}
\mathrm{RHS}=\inf\left\{c>0\;\left|\;\left(\exists_\infty(p,q)\in(\mathbb{Z}\backslash\{0\})^2\right)\left(p(q-p\theta)<c\right)\right.\right\}.
\end{equation}
 % -------------- %
Our goal is to prove that $\mathrm{LHS}=\mathrm{RHS}$. To that end we will use the identity
 % -------------- %
\begin{equation}\label{identita}
q(q\theta^{-1}-p)=p(q-p\theta)+\frac{[p(q-p\theta)]^2}{p^2\theta}\,,
\end{equation}
 % -------------- %
which implies
 % -------------- %
$$
q(q\theta^{-1}-p)\geq p(q-p\theta) \qquad\mbox{for all} \;\; (p,q)\in(\mathbb{Z}\backslash\{0\})^2\,;
$$
 % -------------- %
hence $\mathrm{LHS}\geq\mathrm{RHS}$. At the same time, for every $c>\mathrm{RHS}$ there are infinitely many $(p,q)\in(\mathbb{Z}\backslash\{0\})^2$ such that $p(q-p\theta)<c$. Therefore, due to identity (\ref{identita}), there are infinitely many $(p,q)\in(\mathbb{Z}\backslash\{0\})^2$ such that
 % -------------- %
$$
q(q\theta^{-1}-p)<c+\frac{c^2}{p^2\theta}\,.
$$
 % -------------- %
Choosing $p$ large enough, we can find for any $c>\mathrm{RHS}$ and any $\epsilon>0$ infinitely many pairs $(p,q)\in\mathbb{Z}^2$ with the property
 % -------------- %
$$
q(q\theta^{-1}-p)<c+\epsilon\,.
$$
 % -------------- %
Consequently, we have also the inequality $\mathrm{LHS}\leq\mathrm{RHS}$ which completes the proof of the sought relation $\mathrm{LHS}=\mathrm{RHS}$.

It remains to prove formula (\ref{upsilon iii}). We know from the previous step that $\mu(\theta^{-1})=\mathrm{RHS}$ according to (\ref{RHS}), hence
 % -------------- %
\begin{equation}\label{upsilon(theta-1)}
\upsilon(\theta^{-1})=\inf\left\{c>0\;\left|\;\left(\exists_\infty(p,q)\in\mathbb{Z}^2\right)\left(\frac{q}{p}-\theta<\frac{c}{p^2}\right)\right.\right\}.
\end{equation}
 % -------------- %
Formula (\ref{upsilon iii}) follows trivially from equations (\ref{upsilon}), (\ref{upsilon(theta-1)}) (where we have to rename the variables,  $p\mapsto q$, $q\mapsto p$) and (\ref{mu(theta)}).
\end{proof}
 % -------------- %
Regarding equation (\ref{upsilon iii}), let us remark that the values $\upsilon(\theta)$ and $\upsilon(\theta^{-1})$ may or may not coincide. For example, for the golden mean, $\phi=(\sqrt{5}+1)/2$, we have $\upsilon(\phi)=\upsilon(\phi^{-1})=1/\sqrt{5}$ (see Section~\ref{s:golden} below), on the other hand, the literature on the Markov constant provides hints of the existence of numbers $\theta$ with the property $\upsilon(\theta)\neq\upsilon(\theta^{-1})$, see e.g.\ \cite{PSZ16}.

\medskip

Function $\upsilon(\theta)$ is closely related to approximations of $\theta$ by rationals. A number $\frac{p}{q}\in\mathbb{Q}$ with $p,q\in\mathbb{Z}$ is called \emph{best Diophantine approximation of the second kind} to a given $\theta\in\mathbb{R}$ if
 % -------------- %
\begin{equation}\label{Diophantine}
|q\theta-p|<|q'\theta-p'|
\end{equation}
 % -------------- %
holds for all $\frac{p'}{q'}\neq\frac{p}{q}$ such that $p',q'\in\mathbb{Z}$ and $0<q'\leq q$. A $\frac{p}{q}$ is a best Diophantine approximation of the second kind to a $\theta\in\mathbb{R}$ if and only if it is a convergent of the continued fraction corresponding to $\theta$ (except for the trivial case $\theta=a_0+\frac{1}{2}$, $\frac{p}{q}=\frac{a_0}{1}$), see e.g.~\cite{Kh64}. Recall that a \emph{convergent} of a $\theta\in\mathbb{R}$ is a rational number equal to a finite initial segment of a continued fraction representation of $\theta$, e.g., $a_0$, $a_0+\frac{1}{a_1}$, $a_0+\frac{1}{a_1+\frac{1}{a_2}}$, etc. If the inequality (\ref{Diophantine}) is replaced with $\left|\theta-\frac{p}{q}\right|< \left|\theta -\frac{p'}{q'}\right|$, the corresponding fraction $\frac{p}{q}$ is called \emph{best Diophantine approximation of the first kind} to the number $\theta$.

For the discussion of the problem we address in this work, we will need a certain type of one-sided best approximations, which we will call, in analogy to the notions mentioned above, `best approximation from below (respectively, from above) of the third kind'. They are defined as follows.
 % -------------- %
\begin{definition}\label{Def. third}
Let $\theta\in\mathbb{R}$ and $\frac{p}{q}\in\mathbb{Q}$ for $p,q\in\mathbb{Z}$. We say that the number $\frac{p}{q}$ is a \emph{best approximation from below of the third kind} to $\theta$ if
 % -------------- %
\begin{equation}\label{third below}
0\leq q(q\theta-p)<q'(q'\theta-p')
\end{equation}
 % -------------- %
for all $\frac{p'}{q'}\geq\theta$ such that $\frac{p'}{q'}\neq\frac{p}{q}$, $p',q'\in\mathbb{Z}$ and $0<q'\leq q$. Likewise, we call $\frac{p}{q}$ a \emph{best approximation from above of the third kind} to $\theta$ if
 % -------------- %
\begin{equation}\label{third above}
0\leq q(p-q\theta)<q'(p'-q'\theta)
\end{equation}
 % -------------- %
for all $\frac{p'}{q'}\leq\theta$ such that $\frac{p'}{q'}\neq\frac{p}{q}$, $p',q'\in\mathbb{Z}$ and $0<q'\leq q$.
\end{definition}
 % -------------- %

\medskip

Notice that for every irrational $\theta$, there are infinitely many best approximations from below of the third kind to $\theta$. Let us regard them as a sequence $(\frac{p_n}{q_n})_{n=0}^{\infty}$ such that $q_n$ grow with $n$. According to Definition~\ref{Def. third}, the corresponding values $q_n(q_n\theta-p_n)$ form a decreasing sequence, the limit of which cannot be less than $\upsilon(\theta)$ by (\ref{upsilon}). Furthermore, considering (\ref{upsilon}) and the fact that the approximations $\frac{p_n}{q_n}$ are \emph{best} in terms of (\ref{third below}), we have the result
 % -------------- %
\begin{equation}\label{from below upsilon}
\upsilon(\theta)=\lim_{n\to\infty} q_n(q_n\theta-p_n).
\end{equation}
 % -------------- %
Formula~(\ref{from below upsilon}) combined with an explicit characterization of the terms $\frac{p_n}{q_n}$, which will be found in Proposition~\ref{Prop. third from below} below, greatly simplifies the evaluation of the function $\upsilon(\theta)$.

\begin{lemma}\label{Lem. third nonconv}
Let $\theta=[a_0;a_1,a_2,a_3,\ldots]$ and $\frac{p_n}{q_n},\;n\in\mathbb{N}$, be convergents of $\theta$. If the inequalities
 % -------------- %
$$
\frac{p_{n-1}}{q_{n-1}}<\frac{p}{q}<\frac{p_{n+1}}{q_{n+1}}\leq\theta \qquad\text{or}\qquad \frac{p_{n-1}}{q_{n-1}}>\frac{p}{q}>\frac{p_{n+1}}{q_{n+1}}\geq\theta
$$
 % -------------- %
hold, then we have
 % -------------- %
$$
q|q\theta-p|>\frac{1}{a_n}\,.
$$
 % -------------- %
\end{lemma}
 % -------------- %
\begin{proof}
First we estimate the absolute value $\left|\frac{p}{q}-\frac{p_{n-1}}{q_{n-1}}\right|$ from below,
 % -------------- %
\begin{equation}\label{p/q n-1}
\left|\frac{p}{q}-\frac{p_{n-1}}{q_{n-1}}\right|=\frac{|pq_{n-1}-qp_{n-1}|}{q\cdot q_{n-1}}\geq\frac{1}{q\cdot q_{n-1}}\,,
\end{equation}
 % -------------- %
where we used a trivial fact that $|pq_{n-1}-qp_{n-1}|\geq1$ because the expression is by assumption a nonzero integer. In the next step we find an upper estimate of the same quantity, taking advantage of a known formula $\frac{p_{k-2}}{q_{k-2}}-\frac{p_k}{q_k}=\frac{(-1)^{k-1}a_k}{q_kq_{k-2}}\:$ (cf.~\cite[Cor.~of Thm.~3]{Kh64}) for $[a_0;a_1,a_2,\ldots]$ representing the continued-fraction form of $\theta$,
 % -------------- %:
\begin{equation}\label{n-1 n+1}
\left|\frac{p}{q}-\frac{p_{n-1}}{q_{n-1}}\right|<\left|\frac{p_{n+1}}{q_{n+1}}-\frac{p_{n-1}}{q_{n-1}}\right|=\frac{a_n}{q_{n+1}q_{n-1}}\,.
\end{equation}
 % -------------- %
Combining inequalities (\ref{p/q n-1}) and (\ref{n-1 n+1}), we obtain
 % -------------- %
\begin{equation}\label{q a_n}
q>\frac{q_{n+1}}{a_n}\,.
\end{equation}
 % -------------- %
Now we use the assumptions of the lemma to estimate $\left|\theta-\frac{p}{q}\right|$:
 % -------------- %
\begin{equation*}
\left|\theta-\frac{p}{q}\right|\geq\left|\frac{p_{n+1}}{q_{n+1}}-\frac{p}{q}\right|=\frac{|qp_{n+1}-pq_{n+1}|}{q\cdot q_{n-1}}\geq\frac{1}{q\cdot q_{n+1}}\,.
\end{equation*}
 % -------------- %
Hence we obtain, taking advantage of inequality (\ref{q a_n}),
 % -------------- %
\begin{equation*}
q|q\theta-p|\geq\frac{q}{q_{n+1}}>\frac{1}{a_n}\,,
\end{equation*}
 % -------------- %
which yields the sought claim.
\end{proof}
 % -------------- %

\medskip

 % -------------- %
\begin{proposition}\label{Prop. third from below}
Every best approximation of the third kind from below to a number $\theta\in\mathbb{R}$ is a convergent of $\theta$.
\end{proposition}
 % -------------- %
\begin{proof}
We will proceed by \emph{reductio ad absurdum}. Suppose that $\frac{p}{q}$ is a best approximation of the third kind from below of $\theta$ which is not a convergent of $\theta$. Then either we have $\frac{p}{q}<\frac{p_0}{q_0}=\lfloor\theta\rfloor$, where $\lfloor\cdot\rfloor$ is the floor function, or $\frac{p}{q}$ lies between two convergents that are smaller or equal to $\theta$. First we will disprove the former case. For every $\frac{p}{q}<\lfloor\theta\rfloor$ we have
 % -------------- %
$$
q(q\theta-p)=q^2\left(\theta-\frac{p}{q}\right)\geq\theta-\frac{p}{q}>\theta-\lfloor\theta\rfloor=1\cdot(1\cdot\theta-\lfloor\theta\rfloor).
$$
 % -------------- %
Comparing this result with condition (\ref{third below}) for $p'=\lfloor\theta\rfloor$ and $q'=1$, we see that $\frac{p}{q}$ cannot be a best approximation from below of the third kind.

In the rest of the proof we will therefore suppose that $\frac{p}{q}$ lies between two convergents that are smaller or equal to $\theta$, i.e.
 % -------------- %
\begin{equation}\label{bounds p/q}
\frac{p_{n-1}}{q_{n-1}}<\frac{p}{q}<\frac{p_{n+1}}{q_{n+1}}\leq\theta \;\quad\mathrm{\;\;for\;a  \;certain\;odd}\;\:n\,;
\end{equation}
 % -------------- %
recall that the parity of $n$ determines whether the convergents are larger or smaller than~$\theta$. Our goal is to show that $\frac{p}{q}$ contradicts the requirement (\ref{third below}) on a best approximation of the third kind from below, for which it suffices to demonstrate that
 % -------------- %
\begin{equation}\label{contradiction}
q>q_{n-1} \quad\wedge\quad q(q\theta-p)\geq q_{n-1}(q_{n-1}\theta-p_{n-1}).
\end{equation}
 % -------------- %
On one hand, obviously
 % -------------- %
\begin{equation}\label{ineq1 q}
\frac{p}{q}-\frac{p_{n-1}}{q_{n-1}}=\frac{pq_{n-1}-qp_{n-1}}{q\cdot q_{n-1}}\geq\frac{1}{q\cdot q_{n-1}}\,.
\end{equation}
 % -------------- %
On the other hand, the well-known formula $\frac{p_{k}}{q_{k}}-\frac{p_{k-1}}{q_{k-1}}=\frac{(-1)^{k+1}}{q_{k}q_{k-1}}$ in combination with assumptions (\ref{bounds p/q}) implies
 % -------------- %
\begin{equation}\label{ineq2 q}
\hspace{-6.2em} \frac{p}{q}-\frac{p_{n-1}}{q_{n-1}}<\frac{p_{n+1}}{q_{n+1}}-\frac{p_{n-1}}{q_{n-1}}
=\frac{p_{n+1}}{q_{n+1}}-\frac{p_n}{q_n}+\frac{p_n}{q_n}-\frac{p_{n-1}}{q_{n-1}}=
\frac{(-1)^n}{q_{n+1}q_{n}}+\frac{(-1)^{n-1}}{q_nq_{n-1}}=\frac{q_{n+1}-q_{n-1}}{q_{n-1}q_{n}q_{n+1}}.
\end{equation}
 % -------------- %
Combining inequalities (\ref{ineq1 q}) and (\ref{ineq2 q}), we obtain
 % -------------- %
$$
q>\frac{q_{n}q_{n+1}}{q_{n+1}-q_{n-1}}\,,
$$
 % -------------- %
which, in particular, implies $q>q_n$. Consequently,
 % -------------- %
\begin{equation}\label{q>q_n-1}
q>q_n>q_{n-1}\,.
\end{equation}
 % -------------- %
This verifies the first part of (\ref{contradiction}). In the next step we estimate $q_{n-1}(q_{n-1}\theta-p_{n-1})$. Since $\frac{p_{n-1}}{q_{n-1}}$ is a convergent, we have
 % -------------- %
$$
\theta-\frac{p_{n-1}}{q_{n-1}}<\frac{1}{q_{n}q_{n-1}}
$$
 % -------------- %
or, in other words
 % -------------- %
\begin{equation}\label{estimate q n-1}
q_{n-1}(q_{n-1}\theta-p_{n-1})<\frac{q_{n-1}}{q_n}\,.
\end{equation}
 % -------------- %
Now we use Lemma~\ref{Lem. third nonconv} to obtain the estimate
 % -------------- %
\begin{equation}\label{estimate q}
q(q\theta-p)>\frac{1}{a_n}\,.
\end{equation}
 % -------------- %
A well-known rule for continued fractions, $q_n=a_nq_{n-1}+q_{n-2}$, implies $q_n>a_nq_{n-1}$, and therefore
 % -------------- %
\begin{equation}\label{comparison}
\frac{1}{a_n}>\frac{q_{n-1}}{q_n}\,.
\end{equation}
 % -------------- %
Inequalities (\ref{estimate q n-1}), (\ref{estimate q}) and (\ref{comparison}) together imply $q(q\theta-p)>q_{n-1}(q_{n-1}\theta-p_{n-1})$. Taking into account that $q>q_{n-1}$, in view of estimate (\ref{q>q_n-1}), we conclude that $\frac{p}{q}$ is not a best approximation of the third kind from below to $\theta$.
\end{proof}

Let us remark that \emph{not every} convergent $\frac{p}{q}<\theta$ is a best approximation of the third kind from below to $\theta$. For example, $\frac{333}{106}$ is a convergent of $\pi$, but does not obey the condition~(\ref{third below}) (which can be checked by considering $\frac{p'}{q'}=\frac{3}{1}$).

\medskip

As for the approximation from above, the situation is slightly different.
 % -------------- %
\begin{proposition}\label{Prop. third from above}
Every best approximation from above of the third kind to a $\theta\in\mathbb{R}$ is either $\lceil\theta\rceil$ or a convergent of $\theta$.
\end{proposition}
 % -------------- %
\begin{proof}
We proceed again by contradiction. Let $\frac{p}{q}\neq\lceil\theta\rceil$ be a best approximation of the third kind from above to $\theta$ which is not a convergent of $\theta$. Then either $\frac{p}{q}$ lies between two convergents that are smaller than $\theta$, or $\frac{p}{q}>\frac{p_1}{q_1}\wedge\frac{p}{q}\neq\lceil\theta\rceil$. The former case can be treated in the same manner as in the proof of Proposition~\ref{Prop. third from below}; therefore, we will omit it here and proceed directly to the case $\frac{p}{q}>\frac{p_1}{q_1}$, $\frac{p}{q}\neq\lceil\theta\rceil$. Since $\frac{p_1}{q_1}=a_0+\frac{1}{a_1}=\frac{a_0a_1+1}{a_1}$, every $\frac{p}{q}>\frac{p_1}{q_1}$ satisfies
 % -------------- %
\begin{equation}\label{p}
p>q\left(a_0+\frac{1}{a_1}\right)=qa_0+\frac{q}{a_1}\,.
\end{equation}
 % -------------- %
We distinguish two cases.
 % -------------- %
\begin{itemize}
\item If $q<a_1$, inequality (\ref{p}) gives $p\geq qa_0+1$; hence
 % -------------- %
$$
q(p-q\theta)\geq q(qa_0+1-q\theta)=a_0+1-\theta+(q-1)\left(1-(q+1)(\theta-a_0)\right)
$$
 % -------------- %
(the last equality can be easily checked). The assumption $q<a_1$ gives $q+1\leq a_1$. Taking advantage of the trivial estimate $\theta-a_0\leq\frac{1}{a_1}$, we get $1-(q+1)(\theta-a_0)\geq0$; hence
 % -------------- %
$$
a_0+1-\theta+(q-1)\left(1-(q+1)(\theta-a_0)\right)\geq a_0+1-\theta.
$$
 % -------------- %
Since $a_0+1\geq\lceil\theta\rceil$, we conclude that
 % -------------- %
$$
q(p-q\theta)\geq1\cdot(\lceil\theta\rceil-1\cdot\theta),
$$
 % -------------- %
i.e., every $\frac{p}{q}\neq\lceil\theta\rceil$ contradicts the condition (\ref{third above}) with the choice $p'=\lceil\theta\rceil$, $q'=1$.

\item If $q\geq a_1$, inequality (\ref{p}) gives
 % -------------- %
$$
q(p-q\theta)>q\left(qa_0+\frac{q}{a_1}-q\theta\right)=a_1(a_0a_1+1-a_1\theta)+(q^2-a_1^2)\left(\frac{1}{a_1}-(\theta-a_0)\right).
$$
 % -------------- %
Using the assumption $q\geq a_1$ together with the trivial estimate $\theta-a_0\leq\frac{1}{a_1}$, we get
 % -------------- %
$$
q(p-q\theta)>a_1(a_0a_1+1-a_1\theta);
$$
 % -------------- %
i.e., $\frac{p}{q}$ contradicts the condition (\ref{third above}) with the choice $p'=a_0a_1+1$, $q'=a_1$.
\end{itemize}
 % -------------- %
To sum up, in both cases we found that $\frac{p}{q}>\frac{p_1}{q_1}$, $\frac{p}{q}\neq\lceil\theta\rceil$ cannot be a best approximation from above of the third kind to $\theta$.
\end{proof}

%%%%%%%%%%%%%%%%%%%%%%%%%%%%%%%%%%%%%%%%%%%%%%%%%%%%%%%%%%%%%%%%%%%%%%%%%%%%%%%%%%%%%%%%%%%%%%%%%%%%%%%%%%%%%%%
%%%%------------------------------------------- Number of gaps --------------------------------------------%%%%
%%%%%%%%%%%%%%%%%%%%%%%%%%%%%%%%%%%%%%%%%%%%%%%%%%%%%%%%%%%%%%%%%%%%%%%%%%%%%%%%%%%%%%%%%%%%%%%%%%%%%%%%%%%%%%%
%  Section
%
\section{Number of spectral gaps of lattice graphs}\label{s:gapnumber}
\setcounter{equation}{0}

Now we can address our second main topic, the existence of graphs with the Bethe--Sommerfeld property. As indicated in the introduction, to this aim we shall revisit the model introduced in \cite{Ex95} and further discussed in \cite{Ex96, EG96}. Let us first recall some needed notions. Consider a rectangular lattice graph in the plane with edges of lengths $a$ and $b$ -- cf.~Figure~\ref{fig:lattice}.
 % -------------- %
\begin{figure}[tbp]
    \centering
    \hspace{-10.5em}\includegraphics[width=.4\columnwidth]{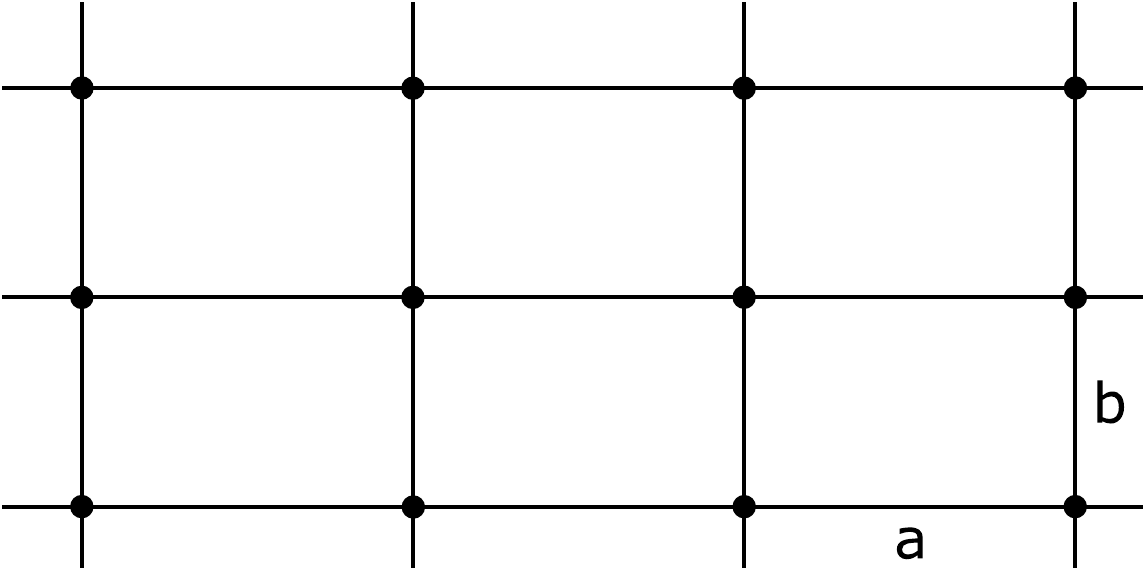}
    \caption{The rectangular-lattice graph}
    \label{fig:lattice}
\end{figure}
 % -------------- %
In addition, suppose that the graph Hamiltonian $H$ is the Laplacian defined as a self-adjoint operator by imposing at each graph vertex $v$ the $\delta$ coupling condition -- that is, continuity together with the requirement $\sum_{j=1}^4 \psi'(v) = \alpha\psi(v)$ -- with a parameter $\alpha\in\R$. According to \cite{Ex96}, a number $k^2>0$ belongs to a gap if and only if $k>0$ satisfies the gap condition, which reads
 % -------------- %
\begin{equation}\label{alpha>0}
\tan\left(\frac{ka}{2}-\frac{\pi}{2}\left\lfloor\frac{ka}{\pi}\right\rfloor\right)+\tan\left(\frac{kb}{2}-\frac{\pi}{2}\left\lfloor\frac{kb}{\pi}\right\rfloor\right)<\frac{\alpha}{2k} \;\quad \text{for}\; \alpha>0
\end{equation}
 % -------------- %
and
 % -------------- %
\begin{equation}\label{alpha<0}
\cot\left(\frac{ka}{2}-\frac{\pi}{2}\left\lfloor\frac{ka}{\pi}\right\rfloor\right)+\cot\left(\frac{kb}{2}-\frac{\pi}{2}\left\lfloor\frac{kb}{\pi}\right\rfloor\right)<\frac{|\alpha|}{2k} \;\quad \text{for}\; \alpha<0\,;
\end{equation}
 % -------------- %
we neglect the case $\alpha=0$ where the spectrum is trivial, $\sigma(H)=[0,\infty)$. Note that for $\alpha<0$ the spectrum extends to the negative part of the real axis and may have a gap there. From the point of view of our present problem this is not that important, though, the reason is that if such a gap exists, it always extends to positive values of the energy -- see Proposition~\ref{Prop general negative} below and Figure~2 in \cite{EG96} -- hence it is sufficient to analyze solutions to the gap conditions (\ref{alpha>0}) and (\ref{alpha<0}) only. Since the sign of $\alpha$ plays role here, it is reasonable to discuss the two cases separately.

Before proceeding further, let us remark that in order to find solutions to conditions (\ref{alpha>0}) and (\ref{alpha<0}) we employ the number-theoretic results of the previous section. This will provide the sought result, in particular, a proof of Theorem~\ref{thm:existence}, but without a convincing insight into the mechanism of the effect. For a comment on that point, representing at the same time a challenge for future work, see Section~\ref{s:concl}.

\subsection{The case $\alpha>0$}

Let us first make the gap description more specific.
 % -------------- %
\begin{proposition}\label{Prop general}
Let $\theta=\frac{a}{b}$. The following claims are valid:
 % -------------- %
\begin{itemize}
\item Every gap in the spectrum has the left (lower) endpoint equal to $k^2=\left(\frac{m\pi}{a}\right)^2$ or $k^2=\left(\frac{m\pi}{b}\right)^2$ for some $m\in\mathbb{N}$.
\item A gap with the left endpoint at $k^2=\left(\frac{m\pi}{a}\right)^2$ is present if and only if
 % -------------- %
\begin{equation}\label{gap F a}
\frac{2m\pi}{a}\tan\left(\frac{\pi}{2}(m\theta^{-1}-\lfloor m\theta^{-1}\rfloor)\right)<\alpha.
\end{equation}
 % -------------- %
\item A gap with the left endpoint at $k^2=\left(\frac{m\pi}{b}\right)^2$ is present if and only if
 % -------------- %
\begin{equation}\label{gap F b}
\frac{2m\pi}{b}\tan\left(\frac{\pi}{2}(m\theta-\lfloor m\theta\rfloor)\right)<\alpha.
\end{equation}
 % -------------- %
\end{itemize}
 % -------------- %
\end{proposition}
 % -------------- %
\begin{proof}
The gap condition (\ref{alpha>0}) is equivalent to $F(k)<\alpha$, where
 % -------------- %
$$
F(k)=2k\left(\tan\left(\frac{ka}{2}-\frac{\pi}{2}\left\lfloor\frac{ka}{\pi}\right\rfloor\right)+\tan\left(\frac{kb}{2}-\frac{\pi}{2}\left\lfloor\frac{kb}{\pi}\right\rfloor\right)\right).
$$
 % -------------- %
Function $k\mapsto F(k)$ has discontinuities at points $k=\frac{m\pi}{a}$ and $k=\frac{m\pi}{b}$ for $m\in\mathbb{N}$. It is easy to check that $F(\cdot)$ is strictly increasing in each interval of continuity and has limits
 % -------------- %
$$
\lim_{k\nearrow\frac{m\pi}{a}}F(k)=\lim_{k\nearrow\frac{m\pi}{b}}F(k)=+\infty
$$
 % -------------- %
at the right endpoints of the continuity intervals. Hence there is at most one gap in each interval of continuity of $F(k)$, and moreover, all gaps are adjacent to points $k^2$ corresponding to $k$ being left endpoints of those intervals. This proves the first part of the proposition.

Furthermore, a gap with the left endpoint equal to $k^2=\left(\frac{m\pi}{a}\right)^2$ is present if and only if $\lim_{k\searrow\frac{m\pi}{a}}F(k)<\alpha$, and since
 % -------------- %
$$
\lim_{k\searrow\frac{m\pi}{a}}F(k)=\frac{2m\pi}{a}\tan\left(\frac{\pi}{2}\left(m\frac{b}{a}-\left\lfloor m\frac{b}{a}\right\rfloor\right)\right)\,,
$$
 % -------------- %
we arrive at the gap conditions (\ref{gap F a}); the gap condition (\ref{gap F b}) is obtained similarly by considering $\lim_{k\searrow\frac{m\pi}{b}}F(k)<\alpha$.
\end{proof}
 % -------------- %
\begin{corollary}\label{Corol none}
Let $\theta=\frac{a}{b}$. If
 % -------------- %
\begin{equation}\label{none}
\hspace{-4em}\frac{2m\pi}{a}\tan\left(\frac{\pi}{2}(m\theta^{-1}-\lfloor m\theta^{-1}\rfloor)\right)\geq\alpha \quad\wedge\quad \frac{2m\pi}{b}\tan\left(\frac{\pi}{2}(m\theta-\lfloor m\theta\rfloor)\right)\geq\alpha
\end{equation}
 % -------------- %
holds for all $m\in\mathbb{N}$, then there are no gaps in the spectrum.
\end{corollary}

\medskip

Next we relate the number of gaps to values of the function $\upsilon(\theta)$ introduced above.
 % -------------- %
\begin{proposition}\label{Prop finite}
Let $\theta=\frac{a}{b}$. If
 % -------------- %
\begin{equation}\label{tan}
\alpha<\pi^2\cdot\min\left\{\frac{\upsilon(\theta)}{b},\frac{\upsilon(\theta^{-1})}{a}\right\},
\end{equation}
then the number of gaps in the spectrum is at most finite.
\end{proposition}
 % -------------- %
\begin{proof}
The expression at the left-hand side of condition (\ref{gap F b}) satisfies
 % -------------- %
$$
\frac{2m\pi}{b}\tan\left(\frac{\pi}{2}\left(\theta m-\lfloor\theta m\rfloor\right)\right)>
\frac{2m\pi}{b}\cdot\frac{\pi}{2}\left(\theta m-\lfloor\theta m\rfloor\right)=\frac{m\pi^2}{b}\cdot\left(\theta m-\lfloor\theta m\rfloor\right).
$$
 % -------------- %
At the same time, (\ref{upsilon i}) implies that for every $c<\upsilon(\theta)$, the inequality
 % -------------- %
$$
\theta m-\lfloor\theta m\rfloor\geq\frac{c}{m}
$$
 % -------------- %
holds except possibly for finitely many values of $m$. Therefore, if $c<\upsilon(\theta)$, we have
 % -------------- %
$$
\frac{2m\pi}{b}\tan\left(\frac{\pi}{2}\left(\theta m-\lfloor\theta m\rfloor\right)\right)>
\frac{m\pi^2}{b}\cdot\frac{c}{m}=\frac{\pi^2}{b}\,c
 % -------------- %
$$
for all $m$ with at most finitely many exceptions. To sum up, if
 % -------------- %
\begin{equation}\label{m fingap b}
\left(\exists c<\upsilon(\theta)\right)\left(\alpha\leq\frac{\pi^2}{b}\,c\right),
\end{equation}
 % -------------- %
the gap condition (\ref{gap F b}) is satisfied for at most finitely many values $m$ only; note that condition (\ref{m fingap b}) is equivalent to
 % -------------- %
\begin{equation}\label{m fingap b sim}
\alpha<\frac{\pi^2}{b}\,\upsilon(\theta)\,.
\end{equation}
 % -------------- %
One can repeat the same considerations for the gap condition (\ref{gap F a}). We get
 % -------------- %
$$
\frac{2m\pi}{a}\tan\left(\frac{\pi}{2}\left(\theta^{-1} m-\lfloor\theta^{-1} m\rfloor\right)\right)>
\frac{2m\pi}{a}\cdot\frac{\pi}{2}\left(\theta^{-1} m-\lfloor\theta^{-1} m\rfloor\right)=\frac{m\pi^2}{a}\left(\theta^{-1} m-\lfloor\theta^{-1} m\rfloor\right).
$$
 % -------------- %
For every $c<\upsilon(\theta^{-1})$ we have in view of (\ref{upsilon i})
 % -------------- %
$$
\theta^{-1} m-\lfloor\theta^{-1} m\rfloor\geq\frac{c}{m}
$$
 % -------------- %
except possibly for finitely many values of $m$. Hence
 % -------------- %
$$
\frac{2m\pi}{a}\tan\left(\frac{\pi}{2}\left(\theta^{-1} m-\lfloor\theta^{-1} m\rfloor\right)\right)>
\frac{m\pi^2}{a}\cdot\frac{c}{m}=\frac{\pi^2}{a}\,c
$$
 % -------------- %
holds for all $m$ with possibly finitely many exceptions. To sum up, if
 % -------------- %
\begin{equation}\label{m fingap a}
\left(\exists c<\upsilon(\theta^{-1})\right)\left(\alpha\leq\frac{\pi^2}{a}\,c\right),
\end{equation}
 % -------------- %
then the gap condition (\ref{gap F a}) is satisfied for at most finitely many values $m$ only, and we can again simplify (\ref{m fingap a}) to the form
 % -------------- %
\begin{equation}\label{m fingap a sim}
\alpha<\frac{\pi^2}{a}\,\upsilon(\theta^{-1})\,.
\end{equation}
 % -------------- %
The assumption (\ref{tan}) guarantees the validity of both (\ref{m fingap b sim}) and (\ref{m fingap a sim}), and thus implies the finiteness of the total number of gaps with regard to Proposition~\ref{Prop general}.
\end{proof}

\medskip

To see that the condition on the number of gaps stated in Proposition~\ref{Prop finite} is sharp, consider now the opposite situation.
 % -------------- %
\begin{proposition}\label{Prop infinite}
Let $\theta=\frac{a}{b}$. For all $\alpha$ satisfying
 % -------------- %
\begin{equation*}
\alpha>\pi^2\cdot\min\left\{\frac{\upsilon(\theta)}{b},\frac{\upsilon(\theta^{-1})}{a}\right\}
\end{equation*}
 % -------------- %
the spectrum has infinitely many gaps.
\end{proposition}
 % -------------- %
\begin{proof}
If $\min\left\{\frac{\upsilon(\theta)}{b},\frac{\upsilon(\theta^{-1})}{a}\right\}=\frac{\upsilon(\theta)}{b}$, we set $c=\sqrt{\frac{b\cdot\alpha\cdot\upsilon(\theta)}{\pi^2}}$. Since $\alpha>\pi^2\cdot\frac{\upsilon(\theta)}{b}$, we have $c>\upsilon(\theta)$. For such $c$ and for any $\delta>0$, equation (\ref{upsilon i}) guarantees that
 % -------------- %
\begin{equation}\label{m c delta b}
\left(\exists_\infty m\in\mathbb{N}\right)\left(m\theta-\lfloor m\theta\rfloor<\frac{c}{m}<\frac{2}{\pi}\delta\right),
\end{equation}
 % -------------- %
where the second inequality can be satisfied by taking values $m$ large enough.
Now we use the general fact
 % -------------- %
\begin{equation}\label{tan delta}
(\forall\xi>1)(\exists \delta>0)(\forall x\in(0,\delta))(\tan x<\xi x)\,.
\end{equation}
 % -------------- %
Taking $\xi=\frac{c}{\upsilon(\theta)}$ and the corresponding $\delta$, we use (\ref{m c delta b}) to estimate the left-hand side of the gap condition (\ref{gap F b}) as follows:
 % -------------- %
\begin{equation}\label{gap delta b}
\hspace{-3em} \frac{2m\pi}{b}\tan\left(\frac{\pi}{2}(m\theta-\lfloor m\theta\rfloor)\right)<\frac{2m\pi}{b}\cdot\frac{c}{\upsilon(\theta)}\cdot\frac{\pi}{2}(m\theta-\lfloor m\theta\rfloor)=\frac{\pi^2c^2}{b\cdot\upsilon(\theta)}\,.
\end{equation}
 % -------------- %
Since $\frac{\pi^2c^2}{b\cdot\upsilon(\theta)}=\alpha$, we have established the existence of infinitely many $m\in\mathbb{N}$
satisfying the gap condition (\ref{gap F b}).
Consequently, the total number of spectral gaps is infinite due to Proposition~\ref{Prop general}.

If $\min\left\{\frac{\upsilon(\theta)}{b},\frac{\upsilon(\theta^{-1})}{a}\right\}=\frac{\upsilon(\theta^{-1})}{a}$, we set $c=\sqrt{\frac{a\cdot\alpha\cdot\upsilon(\theta^{-1})}{\pi^2}}$ and proceed similarly as above. Using function $\upsilon(\theta^{-1})$, we establish the existence of infinitely many $m\in\mathbb{N}$ satisfying the gap condition (\ref{gap F a}).
\end{proof}

\medskip

As an immediate consequence of Propositions \ref{Prop general} and \ref{Prop finite}, we obtain a sufficient condition for the graph in question to have the Bethe--Sommerfeld property:
 % -------------- %
\begin{theorem}\label{Thm positive}
Let $\theta=\frac{a}{b}$ and
 % -------------- %
\begin{equation}\label{gamma}
\hspace{-6em}\gamma:=\min\left\{\inf_{m\in\mathbb{N}}\left\{\frac{2m\pi}{a}\tan\left(\frac{\pi}{2}(m\theta^{-1}-\lfloor m\theta^{-1}\rfloor)\right)\right\},\inf_{m\in\mathbb{N}}\left\{\frac{2m\pi}{b}\tan\left(\frac{\pi}{2}(m\theta-\lfloor m\theta\rfloor)\right)\right\}\right\}.
\end{equation}
 % -------------- %
If the coupling constant $\alpha$ satifies
 % -------------- %
\begin{equation}\label{alpha nonzero finite}
\gamma<\alpha<\pi^2\cdot\min\left\{\frac{\upsilon(\theta)}{b},\frac{\upsilon(\theta^{-1})}{a}\right\},
\end{equation}
 % -------------- %
then there is a nonzero and finite number of gaps in the spectrum.
\end{theorem}
 % -------------- %
In Section~\ref{s:sign} we will show that this claim is nonempty by providing an explicit construction of numbers $\theta$ such that the condition (\ref{alpha nonzero finite}) is satisfied for some $\alpha$.

\begin{remark}\label{estimate Markov}
Using equation (\ref{upsilon iii}), we can estimate the quantity $\min\left\{\frac{\upsilon(\theta)}{b},\frac{\upsilon(\theta^{-1})}{a}\right\}$ in terms of the Markov constant of $\theta$; namely:
 % -------------- %
$$
\frac{\mu(\theta)}{\max\{a,b\}}\leq\min\left\{\frac{\upsilon(\theta)}{b},\frac{\upsilon(\theta^{-1})}{a}\right\}\leq\frac{\mu(\theta)}{\min\{a,b\}}\,.
$$
 % -------------- %
Propositions \ref{Prop finite}, \ref{Prop infinite} and Theorem~\ref{Thm positive} can be thus formulated in a weaker way as follows:
 % -------------- %
\begin{itemize}
\item If $\alpha>\frac{\pi^2\mu(\theta)}{\min\{a,b\}}$, the spectrum has infinitely many gaps.
\item If $\alpha<\frac{\pi^2\mu(\theta)}{\max\{a,b\}}$, the spectrum has at most finitely many gaps.
\item If $\gamma<\alpha<\frac{\pi^2\mu(\theta)}{\max\{a,b\}}$ for $\gamma$ given by (\ref{gamma}), there is a nonzero and finite number of gaps in the spectrum.
\end{itemize}
 % -------------- %
\end{remark}

\subsection{The case $\alpha<0$}

In this situation, the gap condition is of the form $G(k)<|\alpha|$, where
 % -------------- %
$$
G(k):=2k\left(\cot\left(\frac{ka}{2}-\frac{\pi}{2}\left\lfloor\frac{ka}{\pi}\right\rfloor\right)
+\cot\left(\frac{kb}{2}-\frac{\pi}{2}\left\lfloor\frac{kb}{\pi}\right\rfloor\right)\right).
$$
 % -------------- %
Using the identity
 % -------------- %
$$
\cot\left(\frac{\pi}{2}(x-\lfloor x\rfloor)\right)
=\tan\left(\frac{\pi}{2}(\lceil x\rceil-x)\right) \;\quad\mathrm{for\;all}\;\; x\notin\mathbb{Z}\,,
$$
 % -------------- %
we can rewrite $G(k)$ for all $k$ except for the points of discontinuity in the form
 % -------------- %
$$
G(k)=2k\left(\tan\left(\frac{\pi}{2}\left(\left\lceil\frac{ka}{\pi}\right\rceil-\frac{ka}{\pi}\right)\right)
+\tan\left(\frac{\pi}{2}\left(\left\lceil\frac{kb}{\pi}\right\rceil-\frac{kb}{\pi}\right)\right)\right)\,,
$$
 % -------------- %
which allows to write the condition in the form more similar to the case $\alpha>0$, the main difference being the swap between the floor and ceiling functions in the arguments. Since the reasoning is completely analogous to the previous case, we limit ourselves to presenting the results omitting the proofs.
 % -------------- %
\begin{proposition}\label{Prop general negative}
Let $\alpha<0$ and $\theta=\frac{a}{b}$, then the following claims are valid:
 % -------------- %
\begin{itemize}
\item Every gap in the spectrum has the right (upper) endpoint equal to $k^2=\left(\frac{m\pi}{a}\right)^2$ or $k^2=\left(\frac{m\pi}{b}\right)^2$ for some $m\in\mathbb{N}$.
\item A gap with the right endpoint at $k^2=\left(\frac{m\pi}{a}\right)^2$ is present if and only if
 % -------------- %
\begin{equation}\label{gap G a}
\frac{2m\pi}{a}\tan\left(\frac{\pi}{2}\left(\lceil m\theta^{-1}\rceil-m\theta^{-1}\right)\right)<|\alpha|.
\end{equation}
 % -------------- %
\item A gap with the right endpoint at $k^2=\left(\frac{m\pi}{b}\right)^2$ is present if and only if
 % -------------- %
\begin{equation}\label{gap G b}
\frac{2m\pi}{b}\tan\left(\frac{\pi}{2}\left(\lceil m\theta\rceil-m\theta\right)\right)<|\alpha|.
\end{equation}
 % -------------- %
\item In particular, if
 % -------------- %
\begin{equation}\label{none negative}
\hspace{-6.2em}\frac{2m\pi}{a}\tan\left(\frac{\pi}{2}\left(\lceil m\theta^{-1}\rceil-m\theta^{-1}\right)\right)\geq|\alpha| \;\;\wedge\;\; \frac{2m\pi}{b}\tan\left(\frac{\pi}{2}\left(\lceil m\theta\rceil-m\theta\right)\right)\geq|\alpha|
\end{equation}
 % -------------- %
for all $m\in\mathbb{N}$, then there are no gaps in the spectrum.
\end{itemize}
 % -------------- %
\end{proposition}

\medskip

\begin{proposition}\label{Prop infinite negative}
Let $\alpha<0$ and $\theta=\frac{a}{b}$. If
 % -------------- %
\begin{equation*}
|\alpha|<\pi^2\cdot\min\left\{\frac{\upsilon(\theta^{-1})}{b},\frac{\upsilon(\theta)}{a}\right\},
\end{equation*}
the number of gaps in the spectrum is at most finite. On the other hand, for $|\alpha|$ greater than the right-hand side of the above inequality, there are infinitely many spectral gaps.
 % -------------- %
\end{proposition}

Note that in case of attractive potential $\alpha<0$, the bound on
$|\alpha|$ in Proposition~\ref{Prop infinite negative} (i.e.,
$\min\{\upsilon(\theta^{-1})/b,\upsilon(\theta)/a\}$) is different
from the bound in case of a repulsive potential, which is equal to
$\min\{\upsilon(\theta^{-1})/a,\upsilon(\theta)/b\}$ (cf.\
Propositions \ref{Prop finite} and \ref{Prop infinite}). However,
the estimates of the bounds in terms of the Markov constant for
$\alpha<0$ are the same as for $\alpha>0$, cf.\ Remark~\ref{estimate
Markov}, namely
 % -------------- %
\begin{equation}\label{estimate Markov negative}
\frac{\mu(\theta)}{\max\{a,b\}}\leq\min\left\{\frac{\upsilon(\theta^{-1})}{b},\frac{\upsilon(\theta)}{a}\right\}\leq\frac{\mu(\theta)}{\min\{a,b\}}\,.
\end{equation}
 % -------------- %

\medskip

\begin{theorem}\label{Thm negative}
Let $\alpha<0$, $\theta=\frac{a}{b}$, and
 % -------------- %
$$
\gamma:=\min\left\{\inf_{m\in\mathbb{N}}\left\{\frac{2m\pi}{a}\tan\left(\frac{\pi}{2}\left(\lceil m\theta^{-1}\rceil-m\theta^{-1}\right)\right)\right\},\inf_{m\in\mathbb{N}}\left\{\frac{2m\pi}{b}\tan\left(\frac{\pi}{2}\left(\lceil m\theta\rceil-m\theta\right)\right)\right\}\right\}.
$$
 % -------------- %
If the coupling constant $\alpha$ satisfies
 % -------------- %
\begin{equation}
\gamma<|\alpha|<\pi^2\cdot\min\left\{\frac{\upsilon(\theta^{-1})}{b},\frac{\upsilon(\theta)}{a}\right\},
\end{equation}
 % -------------- %
there is a nonzero and finite number of gaps in the spectrum.
\end{theorem}

%%%%%%%%%%%%%%%%%%%%%%%%%%%%%%%%%%%%%%%%%%%%%%%%%%%%%%%%%%%%%%%%%%%%%%%%%%%%%%%%%%%%%%%%%%%%%%%%%%%%%%%%%%%%%%%
%%%%-------------------------------------- Existence of BS graphs -----------------------------------------%%%%
%%%%%%%%%%%%%%%%%%%%%%%%%%%%%%%%%%%%%%%%%%%%%%%%%%%%%%%%%%%%%%%%%%%%%%%%%%%%%%%%%%%%%%%%%%%%%%%%%%%%%%%%%%%%%%%
%  Section
%
\section{Example: golden-mean lattice}\label{s:golden}
\setcounter{equation}{0}

The sufficient conditions in Theorems~\ref{Thm positive} and \ref{Thm negative} do not yet solve our problem because it is not obvious whether these statements are not empty. Let us now examine a particular case discussed already in \cite{Ex96, EG96} in which we choose the golden mean, $\phi=\frac{\sqrt{5}+1}{2}$, for the rectangle side ratio $\theta$.

For proving Theorem~\ref{Thm golden} below, we will employ the convergents of $\phi$. The continued fraction representation of $\phi$ is $[1;1,1,1,\ldots]$, and therefore the convergents are of the form
 % -------------- %
\begin{equation}
\frac{F_{n+1}}{F_n}=\frac{p_{n-1}}{q_{n-1}}\,,
\end{equation}
 % -------------- %
where $F_n$ are Fibonacci numbers; recall that
 % -------------- %
$$
F_n=\frac{\phi^n-(-\phi)^{-n}}{\sqrt{5}}\,.
$$
We will also need the values of $\upsilon(\phi)$ and $\upsilon(\phi^{-1})$. It is possible to find them using formula (\ref{from below upsilon}) and Proposition~\ref{Prop. third from below}, but we instead take advantage of known results on the Markov constant. Since $\phi^{-1}=\phi-1$, we have, due to (\ref{upsilon i}),
 % -------------- %
\begin{equation*}
\begin{array}{rl}
\hspace{-3em}\upsilon(\phi^{-1})=&\inf\left\{c>0\;\left|\;\left(\exists_\infty m\in\mathbb{N}\right)\left(m(m(\phi-1)-\lfloor m(\phi-1)\rfloor)<c\right)\right.\right\} \\
=&\inf\left\{c>0\;\left|\;\left(\exists_\infty m\in\mathbb{N}\right)\left(m(m\phi-\lfloor m\phi\rfloor)<c\right)\right.\right\}=\upsilon(\phi)\,.
\end{array}
\end{equation*}
 % -------------- %
Consequently, equation (\ref{upsilon iii}) implies $\upsilon(\phi)=\upsilon(\phi^{-1})=\mu(\phi)$, where the value of $\mu(\phi)$ is known to be equal to $1/\sqrt{5}$, cf. \cite[Chapter I, Thm.~V]{Ca57}. To sum up,
\begin{equation}\label{upsilon gold}
\upsilon(\phi)=\upsilon(\phi^{-1})=\frac{1}{\sqrt{5}}\,.
\end{equation}
 % -------------- %
\begin{theorem}\label{Thm golden}
Let $\frac{a}{b}=\phi=\frac{\sqrt{5}+1}{2}$, then the following claims are valid:
 % -------------- %
\begin{itemize}
\item [(i)] If $\alpha>\frac{\pi^2}{\sqrt{5}a}$ or $\alpha\leq-\frac{\pi^2}{\sqrt{5}a}$, there are infinitely many spectral gaps.
\item [(ii)] If
 % -------------- %
$$
-\frac{2\pi}{a}\tan\left(\frac{3-\sqrt{5}}{4}\pi\right)\leq\alpha\leq\frac{\pi^2}{\sqrt{5}a}\,,
$$
 % -------------- %
there are no gaps in the spectrum.
\item[(iii)] If
 % -------------- %
\begin{equation}\label{golden mean finite gaps}
-\frac{\pi^2}{\sqrt{5}a}<\alpha<-\frac{2\pi}{a}\tan\left(\frac{3-\sqrt{5}}{4}\pi\right),
\end{equation}
 % -------------- %
there is a nonzero and finite number of gaps in the spectrum.
\end{itemize}
 % -------------- %
\end{theorem}
 % -------------- %
\begin{proof} (i) With regard to (\ref{upsilon gold}), the existence of an infinite number of spectral gaps for $\alpha>\frac{\pi^2}{\sqrt{5}a}$ follows immediately from Proposition~\ref{Prop infinite}, for $\alpha<-\frac{\pi^2}{\sqrt{5}a}$ we similarly employ Proposition~\ref{Prop infinite negative}.

\emph{The case $\alpha=-\frac{\pi^2}{\sqrt{5}a}$.} We shall demonstrate that there are infinitely many $m\in\mathbb{N}$ such that the gap condition (\ref{gap G a}), which reads
 % -------------- %
$$
\frac{2m\pi}{a}\tan\left(\frac{\pi}{2}\left(\lceil m\phi^{-1}\rceil-m\phi^{-1}\right)\right)<\frac{\pi^2}{\sqrt{5}a}\,,
$$
 % -------------- %
is satisfied. Choosing $m=F_{n}$ for even $n$ and using the identity $\phi^{-1}=\phi-1$, we can write the gap condition in the form
 % -------------- %
\begin{equation}\label{cond iv}
F_n\tan\left(\frac{\pi}{2}\left(\lceil F_{n}\phi\rceil-F_{n}\phi\right)\right)<\frac{\pi}{2\sqrt{5}}\,.
\end{equation}
 % -------------- %
For even $n$, we have
 % -------------- %
\begin{eqnarray*}
\hspace{-3em}\lefteqn{F_n\phi=\frac{\phi^n-\phi^{-n}}{\sqrt{5}}\,\phi=\frac{\phi^{n+1}-\phi^{-n+1}}{\sqrt{5}}=
\frac{\phi^{n+1}+\phi^{-n-1}}{\sqrt{5}}+\frac{-\phi^{-n-1}-\phi^{-n+1}}{\sqrt{5}}} \\ &&
\hspace{-3em} =\frac{\phi^{n+1}-(-\phi)^{-(n+1)}}{\sqrt{5}}-\frac{\phi+\phi^{-1}}{\sqrt{5}}\,\phi^{-n}=
F_{n+1}-\phi^{-n}\in(F_{n+1}-1,F_{n+1})\,,
\end{eqnarray*}
 % -------------- %
which means that
 % -------------- %
\begin{equation}\label{Fn even}
\lceil F_n\phi\rceil-F_n\phi=F_{n+1}-F_n\phi=\phi^{-n} \qquad\mbox{for even $n$}.
\end{equation}
 % -------------- %
Hence we get,
using the Taylor series of $\tan(x)$,
 % -------------- %
\begin{eqnarray*}
\lefteqn{F_n\tan\left(\frac{\pi}{2}\left(\lceil F_n\phi\rceil-F_n\phi\right)\right)=
\frac{\phi^{n}-\phi^{-n}}{\sqrt{5}}\tan\left(\frac{\pi}{2}\phi^{-n}\right)} \\ &&
=\frac{1}{\sqrt{5}}\left(\phi^{n}-\phi^{-n}\right)\left(\frac{\pi}{2}\phi^{-n}
+\frac{1}{3}\left(\frac{\pi}{2}\right)^3\phi^{-3n}+\frac{2}{15}\left(\frac{\pi}{2}\right)^5\phi^{-5n}+\cdots\right) \\ &&
=\frac{\pi}{2\sqrt{5}}\left(1-\left(1-\frac{1}{3}\cdot\frac{\pi^2}{4}\right)\phi^{-2n}
-\left(\frac{1}{3}\cdot\frac{\pi^2}{4}-\frac{2}{15}\cdot\frac{\pi^4}{16}\right)\phi^{-4n}\cdots\right)\,.
\end{eqnarray*}
 % -------------- %
That is, taking $n$ even leads to the expansion
 % -------------- %
\begin{equation}\label{rozvoj}
\hspace{-3em}F_n\tan\left(\frac{\pi}{2}\left(\lceil F_n\phi\rceil-F_n\phi\right)\right)=\frac{\pi}{2\sqrt{5}}\left(1+\left(\frac{\pi^2}{12}-1\right)\phi^{-2n}+\mathcal{O}(\phi^{-4n})\right).
\end{equation}
 % -------------- %
Since the coefficient $\frac{\pi^2}{12}-1$ at $\phi^{-2n}$ in (\ref{rozvoj}) is negative, condition (\ref{cond iv}) is satisfied for all sufficiently large even $n$. The gap condition (\ref{gap G a}) with $\alpha=-\frac{\pi^2}{\sqrt{5}a}$ is thus satisfied for infinitely many numbers $m=F_n$ with $n$ being even; then Proposition~\ref{Prop general negative} implies the existence of infinitely many gaps.

\smallskip

(ii) We divide the argument into several parts referring to different values of $\alpha$: \\ \emph{The case $\alpha\in\big(0,\frac{\pi^2}{\sqrt{5}a}\big]$.} Using the identity $\phi^{-1}=\phi-1$, we obtain
 % -------------- %
\begin{eqnarray*}
\lefteqn{\frac{2m\pi}{a}\tan\left(\frac{\pi}{2}(m\phi^{-1}-\lfloor m\phi^{-1}\rfloor)\right)\geq
\frac{2m\pi}{a}\left(\frac{\pi}{2}(m\phi^{-1}-\lfloor m\phi^{-1}\rfloor)\right)} \\ &&
=\frac{\pi^2}{a}m\left(m(\phi-1)-\lfloor m(\phi-1)\rfloor\right)=\frac{\pi^2}{a}m\left(m\phi-\lfloor m\phi\rfloor\right),
\end{eqnarray*}
 % -------------- %
and similarly,
 % -------------- %
\begin{equation}\label{ii a}
\hspace{-5em} \frac{2m\pi}{b}\tan\left(\frac{\pi}{2}(m\phi-\lfloor m\phi\rfloor)\right)\geq
\frac{2m\pi}{b}\left(\frac{\pi}{2}(m\phi-\lfloor m\phi\rfloor)\right)=
\frac{\pi^2}{b}m\left(m\phi-\lfloor m\phi\rfloor\right).
\end{equation}
 % -------------- %
In order to disprove the existence of gaps using Corollary~\ref{Corol none}, we shall demonstrate that
 % -------------- %
\begin{equation}\label{ii none}
\hspace{-5em}\frac{\pi^2}{a}m\left(m\phi-\lfloor m\phi\rfloor\right)\geq\frac{\pi^2}{\sqrt{5}a} \quad\wedge\quad \frac{\pi^2}{b}m\left(m\phi-\lfloor m\phi\rfloor\right)\geq\frac{\pi^2}{\sqrt{5}a} \quad\mathrm{for\;all}\;\; m\in\mathbb{N}\,.
\end{equation}
 % -------------- %
With regard to the assumption $a>b$, condition (\ref{ii none}) is equivalent to
 % -------------- %
\begin{equation}\label{ii none 1}
m\left(m\phi-\lfloor m\phi\rfloor\right)\geq\frac{1}{\sqrt{5}} \quad\mathrm{for\;all}\;\; m\in\mathbb{N}\,,
\end{equation}
 % -------------- %
which we are about to prove. We will verify that $m\left(m\phi-p\right)\geq\frac{1}{\sqrt{5}}$ for any $m\in\mathbb{N}$ and $p\in\mathbb{N}_0$. In view of Definition~\ref{Def. third}, it suffices to consider pairs $(p,m)$ such that $\frac{p}{m}$ is a best approximation from below of the third kind to $\phi$. Such approximations are convergents of $\phi$, cf.\ Proposition~\ref{Prop. third from below}. Convergents of $\phi$ that are smaller than $\phi$ are known to be of the form $\frac{F_{n+1}}{F_{n}}$, where $n$ is odd. We obtain
 % -------------- %
\begin{eqnarray*}
\lefteqn{F_n\left(F_n\phi-F_{n+1}\right)
=\frac{\phi^n+\phi^{-n}}{\sqrt{5}}\left(\frac{\phi^n+\phi^{-n}}{\sqrt{5}}\phi-\frac{\phi^{n+1}-\phi^{-(n+1)}}{\sqrt{5}}\right)} \\ &&
=\frac{1}{5}(\phi+\phi^{-1})(1+\phi^{-2n})=\frac{1+\phi^{-2n}}{\sqrt{5}}>\frac{1}{\sqrt{5}}\,,
\end{eqnarray*}
 % -------------- %
i.e., the inequality $m\left(m\phi-p\right)\geq\frac{1}{\sqrt{5}}$ holds true for each best approximation from below of the third kind to $\phi$. Consequently, it holds true for all $\frac{p}{q}<\theta$, in particular, for $p/q=\lfloor m\phi\rfloor/m$. This proves condition (\ref{ii none 1}), hence there are no spectral gaps for $\alpha\in\big(0,\frac{\pi^2}{\sqrt{5}a}\big]$.

\emph{The case $\alpha=0$.} Kirchhoff couplings obviously generate no gaps\footnote{Note that this also means that Theorem~\ref{thm:gaps general} has no implications for the present case, because Kirchhoff condition is scale-invariant and associated with the $\delta$-coupling of the considered model.}, see also \cite{EG96}.

\emph{The case $\alpha\in\big[-\frac{2\pi}{a}\tan\big(\frac{3-\sqrt{5}}{4}\pi\big),0\big)$.} We are going to show that for all $m\in\mathbb{N}$, condition (\ref{none negative}) holds true; then the claim would follow from Proposition~\ref{Prop general negative}. If $m=1$, we have
 % -------------- %
 $$
\frac{2\cdot1\cdot\pi}{a}\tan\left(\frac{\pi}{2}\left(\lceil 1\cdot\phi^{-1}\rceil-1\cdot\phi^{-1}\right)\right)=\frac{2\pi}{a}\tan\left(\frac{\pi}{2}\cdot\frac{3-\sqrt{5}}{2}\right)\geq|\alpha|
$$
 % -------------- %
and
 % -------------- %
$$
\frac{2\cdot1\cdot\pi}{b}\tan\left(\frac{\pi}{2}\left(\lceil 1\cdot\phi\rceil-1\cdot\phi\right)\right)=\frac{2\pi}{b}\tan\left(\frac{\pi}{2}\cdot\frac{3-\sqrt{5}}{2}\right)\geq|\alpha|\,.
$$
 % -------------- %
If $m\geq2$, we use the identity $\phi^{-1}=\phi-1$ to get
 % -------------- %
\begin{eqnarray*}
\lefteqn{\frac{2m\pi}{a}\tan\left(\frac{\pi}{2}\left(\lceil m\phi^{-1}\rceil-m\phi^{-1}\right)\right)
=\frac{2m\pi}{a}\tan\left(\frac{\pi}{2}\left(\lceil m\phi\rceil-m\phi\right)\right)} \\ &&
>\frac{2m\pi}{a}\left(\frac{\pi}{2}\left(\lceil m\phi\rceil-m\phi\right)\right)
=\frac{\pi^2}{a}m\left(\lceil m\phi\rceil-m\phi\right),
\end{eqnarray*}
 % -------------- %
and
 % -------------- %
$$
\frac{2m\pi}{b}\tan\left(\frac{\pi}{2}\left(\lceil m\phi\rceil-m\phi\right)\right)
>\frac{\pi^2}{b}m\left(\lceil m\phi\rceil-m\phi\right)\,.
$$
 % -------------- %
According to condition (\ref{none negative}), we have to check that
 % -------------- %
$$
\min\left\{\frac{\pi^2}{a}m\left(\lceil m\phi\rceil-m\phi\right)\,,\:  \frac{\pi^2}{b}m\left(\lceil m\phi\rceil-m\phi\right)\right\} \geq\frac{2\pi}{a}\tan\left(\frac{3-\sqrt{5}}{4}\pi\right)
$$
 % -------------- %
holds for all $m\ge 2$, which is equivalent, due to $a>b$, to
 % -------------- %
\begin{equation}\label{iii none}
\hspace{-2em} m\left(\lceil m\phi\rceil-m\phi\right)\geq\frac{2}{\pi}\tan\left(\frac{3-\sqrt{5}}{4}\pi\right)\approx0.4355 \quad\mathrm{for\;all}\;\; m\geq2.
\end{equation}
 % -------------- %
Again, in view of Definition~\ref{Def. third}, it is sufficient to verify that $m\left(p-m\phi\right)\geq\frac{2}{\pi}\tan\left(\frac{3-\sqrt{5}}{4}\pi\right)$ holds for $\frac{p}{m}$ (with $m\geq2$) being best approximations from above of the third kind to $\phi$. According to Proposition~\ref{Prop. third from above}, such approximations are convergents of $\phi$, i.e., we have to consider $\frac{p}{q}$ taking the form $\frac{F_{n+1}}{F_{n}}$, where $n$ is even. For this choice we obtain
 % -------------- %
\begin{eqnarray*}
\lefteqn{F_n\left(F_{n+1}-F_n\phi\right)
=\frac{\phi^n-\phi^{-n}}{\sqrt{5}}\left(\frac{\phi^{n+1}+\phi^{-(n+1)}}{\sqrt{5}}-\frac{\phi^n-\phi^{-n}}{\sqrt{5}}\phi\right)} \\ &&
=\frac{1}{5}(\phi+\phi^{-1})(1-\phi^{-2n})=\frac{1-\phi^{-2n}}{\sqrt{5}}\,.
\end{eqnarray*}
 % -------------- %
Moreover, we may assume $n\geq4$, because $F_4=3$ is the smallest Fibonacci number $F_n$ obeying our conditions (having an even index $n$ and satisfying $m=F_n\geq2$). Hence
 % -------------- %
$$
F_n\left(F_{n+1}-F_n\phi\right)\geq\frac{1-\phi^{-8}}{\sqrt{5}}\,,
$$
 % -------------- %
and consequently,
 % -------------- %
$$
m\left(p-m\phi\right)\geq\frac{1-\phi^{-8}}{\sqrt{5}}\approx0.4377
$$
 % -------------- %
for all $\frac{p}{m}>\phi$; in particular, for $p=\lceil m\phi\rceil$. This verifies condition (\ref{iii none}), hence there are no gaps in the spectrum.

\smallskip

(iii) It remains to deal with the case when $-\frac{\pi^2}{\sqrt{5}a}<\alpha<-\frac{2\pi}{a}\tan\left(\frac{3-\sqrt{5}}{4}\pi\right)$. The claim follows from Theorem~\ref{Thm negative} in combination with equation (\ref{upsilon gold}) and the estimate
 % -------------- %
\begin{eqnarray*}
\hspace{-6em}\lefteqn{\inf_{m\in\mathbb{N}}\left\{\frac{2m\pi}{a}\tan\left(\frac{\pi}{2}\left(\lceil m\phi^{-1}\rceil-m\phi^{-1}\right)\right)\right\}\leq\frac{2\cdot1\cdot\pi}{a}\tan\left(\frac{\pi}{2}\left(\lceil 1\cdot\phi^{-1}\rceil-1\cdot\phi^{-1}\right)\right)} \\ &&
\hspace{-6em}=\frac{2\pi}{a}\tan\left(\frac{\pi}{2}\cdot\frac{3-\sqrt{5}}{2}\right).
\end{eqnarray*}
 % -------------- %
This concludes the proof of the theorem.
\end{proof}

In particular, the claim (iii) of Theorem~\ref{Thm golden} provides and affirmative answer to the question we have posed in the introduction.

 % -------------- %
\begin{corollary}
Theorem~\ref{thm:existence} is valid.
\end{corollary}
 % -------------- %
\begin{remark}
Note that a finite nonzero number of gaps in the spectrum can occur only for $\alpha<0$. If $\alpha>0$, there are either no gaps in the spectrum or infinitely many of them in accordance with the numerical observation made in \cite{EG96}. In addition, the window in which the golden-mean lattice has the Bethe--Sommerfeld property is narrow, roughly can be characterized as $4.298 \lesssim -\alpha a \lesssim 4.414$.
\end{remark}
 % -------------- %

We are also able to control the number of gaps in the Bethe--Sommerfeld regime.
 % -------------- %
\begin{theorem}
For a given $N\in\mathbb{N}$, there are exactly $N$ gaps in the spectrum if and only if $\alpha$ is chosen within the bounds
 % -------------- %
\begin{equation}\label{N gaps}
\hspace{-6.5em}-\frac{2\pi\left(\phi^{2(N+1)}-\phi^{-2(N+1)}\right)}{\sqrt{5}a}\tan\left(\frac{\pi}{2}\phi^{-2(N+1)}\right)
\leq\alpha<-\frac{2\pi\left(\phi^{2N}-\phi^{-2N}\right)}{\sqrt{5}a}\tan\left(\frac{\pi}{2}\phi^{-2N}\right)\!.
\end{equation}
 % -------------- %
\end{theorem}
 % -------------- %
\begin{proof}
The bounds on $\alpha$ can be concisely written as $-\frac{A_{N+1}}{a}\leq\alpha<-\frac{A_N}{a}$, where
 % -------------- %
$$
A_{j}:=\frac{2\pi\left(\phi^{2j}-\phi^{-2j}\right)}{\sqrt{5}}\tan\left(\frac{\pi}{2}\phi^{-2j}\right).
$$
 % -------------- %
One can easily check that $\{A_j\}_{j=1}^\infty$ is an increasing sequence with the property
 % -------------- %
$$
A_1=\frac{2\pi\left(\phi^{2}-\phi^{-2}\right)}{\sqrt{5}}\tan\left(\frac{\pi}{2}\phi^{-2}\right)=2\pi\tan\left(\frac{3-\sqrt{5}}{4}\pi\right)
$$
 % -------------- %
and
 % -------------- %
\begin{equation}\label{Aj}
A_j<\frac{\pi^2}{\sqrt{5}} \quad\; \text{for\;all}\;\; j\in\mathbb{N}\,.
\end{equation}
 % -------------- %
Let us examine validity of the conditions (\ref{gap G a}) and (\ref{gap G b}) for $m\in\mathbb{N}$. Using the identity $\phi^{-1}=\phi-1$, we can rewrite them in the form
 % -------------- %
\begin{equation}\label{gap phi a}
\frac{2m\pi}{a}\tan\left(\frac{\pi}{2}\left(\lceil m\phi\rceil-m\phi\right)\right)<|\alpha|
\end{equation}
 % -------------- %
and
 % -------------- %
\begin{equation}\label{gap phi b}
\frac{2m\pi}{b}\tan\left(\frac{\pi}{2}\left(\lceil m\phi\rceil-m\phi\right)\right)<|\alpha|\,,
\end{equation}
 % -------------- %
respectively.

We start with the situation where $m=F_n$ for an even $n$. In this case we have $\lceil F_n\phi\rceil-F_n\phi=\phi^{-n}$, cf.~(\ref{Fn even}). The gap condition (\ref{gap phi a}) for $m=F_n$ with $n$ even thus acquires the form
 % -------------- %
$$
\frac{2\pi\left(\phi^{n}-\phi^{-n}\right)}{\sqrt{5}a}\tan\left(\frac{\pi}{2}\phi^{-n}\right)<|\alpha|\,,
$$
 % -------------- %
in other words, $\frac{1}{a}A_\frac{n}{2}<|\alpha|$. Since $|\alpha|\in\big(\frac{A_N}{a},\frac{A_{N+1}}{a}\big]$ in view of the assumptions (\ref{N gaps}), the gap condition (\ref{gap phi a}) is obviously satisfied with $m=F_n$ for all even values $n=2,4,\ldots,2N$, and violated for even values $n\geq2(N+1)$. Similarly, the gap condition (\ref{gap phi b}) acquires the form
 % -------------- %
$$
\frac{1}{b}A_\frac{n}{2}<|\alpha|.
$$
 % -------------- %
Since
 % -------------- %
$$
\frac{1}{b}A_\frac{n}{2}=\frac{\phi}{a}A_\frac{n}{2}\geq\frac{\phi}{a}A_1
=\phi\frac{2\pi}{a}\tan\left(\frac{3-\sqrt{5}}{4}\pi\right)\approx\frac{6.955}{a}
$$
 % -------------- %
and
 % -------------- %
$$
|\alpha|\leq\frac{\pi^2}{\sqrt{5}a}\approx\frac{4.414}{a}\,,
$$
 % -------------- %
we have $\frac{1}{b}A_\frac{n}{2}\nless|\alpha|$. Consequently, the gap condition (\ref{gap phi b}) cannot be satisfied for the special choice $m=F_n$ with $n$ even.

Let us proceed to the situation when $m$ is different from the values $F_n$ with even indices $n$. In this case we will show that none of the gap conditions (\ref{gap phi a}) and (\ref{gap phi b}) is satisfied. First, we estimate an expression appearing on the left-hand side of conditions (\ref{gap phi a}) and (\ref{gap phi b}) as follows:
 % -------------- %
\begin{eqnarray*}
\hspace{-3em} 2\pi m\tan\left(\frac{\pi}{2}\left(\lceil m\phi\rceil-m\phi\right)\right)
\geq2\pi m\frac{\pi}{2}\left(\lceil m\phi\rceil-m\phi\right)=
\pi^2 m\left(\lceil m\phi\rceil-m\phi\right)\,.
\end{eqnarray*}
 % -------------- %
The bounds (\ref{N gaps}) together with the estimate (\ref{Aj}) imply that $|\alpha|<\frac{\pi^2}{\sqrt{5}a}$. Therefore, conditions (\ref{gap phi a}) and (\ref{gap phi b}) can be disproved for a given $m$ by showing that
 % -------------- %
\begin{equation}\label{none N 2}
\frac{\pi^2}{a}m\left(\lceil m\phi\rceil-m\phi\right)\geq\frac{\pi^2}{\sqrt{5}a} \quad\wedge\quad \frac{\pi^2}{b}m\left(\lceil m\phi\rceil-m\phi\right)\geq\frac{\pi^2}{\sqrt{5}a}\,.
\end{equation}
 % -------------- %
Since $a>b$ holds by assumption, condition (\ref{none N 2}) is equivalent to
 % -------------- %
\begin{equation}\label{none N}
m\left(\lceil m\phi\rceil-m\phi\right)\geq\frac{1}{\sqrt{5}}\,,
\end{equation}
 % -------------- %
which we are now about to prove. We distinguish the following three possibilities:
 % -------------- %
\begin{itemize}
\item[(i)] $\frac{\lceil m\phi\rceil}{m}$ lies between two convergents greater than $\theta$, that is, $\frac{\lceil m\phi\rceil}{m}\in\left(\frac{F_{n+3}}{F_{n+2}},\frac{F_{n+1}}{F_{n}}\right)$ for a certain even $n$;
\item[(ii)] $\frac{\lceil m\phi\rceil}{m}$ lies above the greatest convergent $\frac{F_3}{F_2}=\frac{2}{1}$;
\item[(iii)] $m=r\cdot F_n$ and $\lceil m\phi\rceil=r\cdot F_{n+1}$ holds for a certain $r\geq2$ and even $n\in\mathbb{N}$.
\end{itemize}
 % -------------- %
In case (i) we use Lemma~\ref{Lem. third nonconv} to obtain the estimate
 % -------------- %
$$
m\left(\lceil m\phi\rceil-m\phi\right)>
\frac{1}{a_n}=1\,,
$$
 % -------------- %
which means that (\ref{none N}) holds true. Case (ii) is actually impossible. Indeed,
one can easily check that $\frac{\lceil m\phi\rceil}{m}\leq2$ for all $m\in\mathbb{N}$.
Finally, in case (iii) we get
 % -------------- %
$$
m\left(\lceil m\phi\rceil-m\phi\right)=r^2\cdot F_n(F_{n+1}-F_n\phi)=r^2\cdot\frac{1-\phi^{-2n}}{\sqrt{5}}\,.
$$
 % -------------- %
Since $r\geq2$ and $n\in\mathbb{N}$ is even, we have
 % -------------- %
$$
m(\lceil m\phi\rceil-m\phi)\geq4\cdot\frac{1-\phi^{-4}}{\sqrt{5}}\approx\frac{3.42}{\sqrt{5}}\,,
$$
 % -------------- %
and therefore (\ref{none N}) holds true.
Consequently, the gap conditions (\ref{gap phi a}) and (\ref{gap phi b}) cannot be satisfied in any of the cases (i)--(iii).

To sum up, the assumption (\ref{N gaps}) allows the gap condition (\ref{gap phi a}) to be satisfied for $m=F_n$ with $n=2,4,6,\ldots,2N$, while the gap condition (\ref{gap phi b}) is never satisfied. This implies the existence of exactly $N$ gaps in view of Proposition~\ref{Prop general negative}.
\end{proof}

%%%%%%%%%%%%%%%%%%%%%%%%%%%%%%%%%%%%%%%%%%%%%%%%%%%%%%%%%%%%%%%%%%%%%%%%%%%%%%%%%%%%%%%%%%%%%%%%%%%%%%%%%%%%%%%
%%%%---------------------------------------- Sign independence --------------------------------------------%%%%
%%%%%%%%%%%%%%%%%%%%%%%%%%%%%%%%%%%%%%%%%%%%%%%%%%%%%%%%%%%%%%%%%%%%%%%%%%%%%%%%%%%%%%%%%%%%%%%%%%%%%%%%%%%%%%%
%  Section
%
\section{More on the construction of Bethe--Sommerfeld lattice graphs}\label{s:sign}
\setcounter{equation}{0}

As we have seen in the example discussed in Section~\ref{s:golden}, the Bethe--Sommerfeld property for the special case of golden-mean ratio required an attractive $\delta$ coupling. One may ask whether the Bethe--Sommerfeld behaviour is possible for some other ratios, and whether it can occur for repulsive couplings. In this section we give an affirmative answer to both these questions. First, we present an example of an edge ratio $\theta$ for which the Bethe--Sommerfeld property is valid within a certain range of $\alpha$ for both signs of $\alpha$. Then we introduce an explicit method to construct ratios $\theta$ for which the Bethe--Sommerfeld property of the graph is guaranteed.

Let $\theta=\frac{a}{b}$. Without loss of generality, we may assume $\theta<1$, i.e., $a<b$. If $\alpha>0$, then Theorem~\ref{Thm positive} and Remark~\ref{estimate Markov} imply that the rectangular-lattice Hamiltonian has a nonzero and finite number of gaps in its spectrum whenever there exists an $m_+\in\mathbb{N}$ such that
 % -------------- %
$$
\frac{2m_+\pi}{b}\tan\left(\frac{\pi}{2}(m_+\theta-\lfloor m_+\theta\rfloor)\right)<\alpha<\frac{\pi^2\mu(\theta)}{b}.
$$
 % -------------- %
Similarly, if $\alpha<0$, Theorem~\ref{Thm negative} together with the estimate (\ref{estimate Markov negative}) implies that the Hamiltonian has a nonzero and finite number of gaps in the spectrum whenever there exists an $m_-\in\mathbb{N}$ such that
 % -------------- %
$$
\frac{2m_+\pi}{b}\tan\left(\frac{\pi}{2}(\lceil m_-\theta\rceil-m_-\theta)\right)<|\alpha|<\frac{\pi^2\mu(\theta)}{b}.
$$
 % -------------- %
Therefore, the Hamiltonian has a nonzero and finite number of gaps in the spectrum for some repulsive and attractive potentials whenever conditions (\ref{m+}) and (\ref{m-}) below are satisfied, respectively:
 % -------------- %
\begin{equation}\label{m+}
(\exists m_+\in\mathbb{N})\left(\frac{2m_+}{\pi}\tan\left(\frac{\pi}{2}(m_+\theta-\lfloor m_+\theta\rfloor)\right)<\mu(\theta)\right),
\end{equation}
 % -------------- %
\begin{equation}\label{m-}
(\exists m_-\in\mathbb{N})\left(\frac{2m_-}{\pi}\tan\left(\frac{\pi}{2}(\lceil m_-\theta\rceil-m_-\theta)\right)<\mu(\theta)\right).
\end{equation}
 % -------------- %
As the following Theorem explicitly shows, there exists a $\theta$ such that both conditions (\ref{m+}) and (\ref{m-}) are satisfied at the same time.
 % -------------- %
\begin{theorem}\label{Thm gap for both alpha}
Let the edge ratio be
 % -------------- %
\begin{equation}\label{gap for both alpha}
\theta=\frac{2t^3-2t^2-1+\sqrt{5}}{2(t^4-t^3+t^2-t+1)} \quad\mathrm{for}\;\; t\in\mathbb{N},\; t\geq3\,;
\end{equation}
then there is a nonzero and finite number of gaps in the spectrum for some $\alpha>0$ and for some $\alpha<0$ as well.
\end{theorem}
 % -------------- %
\begin{proof}
The number $\theta$ defined in (\ref{gap for both alpha}) can be written as $\theta=\frac{t\phi+1}{(t^2+1)\phi+t}$ for $\phi=\frac{1+\sqrt{5}}{2}$ being the golden mean. Since $\theta$ is equivalent to $\phi$, cf.\ (\ref{Markov equivalence}), the Markov constant of $\theta$ is $\mu(\theta)=\mu(\phi)=\frac{1}{\sqrt{5}}\approx0.4472$.

It is easy to check that conditions (\ref{m+}) and (\ref{m-}) are satisfied for the choice $m_+=1$ and $m_-=t$ with $t\geq3$, respectively. Indeed,
 % -------------- %
$$
\frac{2\cdot1}{\pi}\tan\left(\frac{\pi}{2}(1\cdot\theta-\lfloor1\cdot\theta\rfloor)\right)=\frac{2}{\pi}\tan\left(\frac{\pi}{2}\cdot\frac{2t^3-2t^2-1+\sqrt{5}}{2(t^4-t^3+t^2-t+1)}\right)
$$
 % -------------- %
is a decreasing function of $t$ that has an approximate value $0.3310<\mu(\theta)$ at $t=3$. Similarly, for $m_-=t$, we get
 % -------------- %
$$
\frac{2t}{\pi}\tan\left(\frac{\pi}{2}(\lceil t\theta\rceil-t\theta)\right)=\frac{2t}{\pi}\tan\left(\frac{\pi}{2}\cdot\frac{2t^2-t-t\sqrt{5}+2}{2(t^4-t^3+t^2-t+1)}\right)\,,
$$
 % -------------- %
which is again a decreasing function of $t$ being approximately equal to $0.2546<\mu(\theta)$ at the point $t=3$.
\end{proof}

Let us proceed to a general method to construct ratios $\theta$ that give rise to graphs with the Bethe--Sommerfeld property. We start from any badly approximable irrational number $\beta\in(0,1)$ with a continued-fraction representation
 % -------------- %
$$
\beta=[0;c_1,c_2,c_3,\ldots]\,;
$$
 % -------------- %
recall that $\beta$ is badly approximable if and only if the terms $c_1,c_2,c_3,\ldots$ are bounded.
Then we define numbers $\rho$, $\varsigma$ and $\tau$ with continued-fraction representations
 % -------------- %
\begin{eqnarray}
\rho=&[0;t,c_1,c_2,c_3,\ldots]\,; \label{constr t} \\
\varsigma=&[0;1,t,c_1,c_2,c_3,\ldots]\,; \label{constr 1t} \\
\tau=&[0;t,t,c_1,c_2,c_3,\ldots] \label{constr tt}
\end{eqnarray}
 % -------------- %
for $t\in\mathbb{N}$ being a parameter to be specified. Since the numbers $\rho,\varsigma,\tau$ are equivalent to $\beta$, cf.~(\ref{equiv. cont. fr.}), we have
 % -------------- %
$$
\mu(\rho)=\mu(\varsigma)=\mu(\tau)=\mu(\beta)\,,
$$
 % -------------- %
where $\mu(\beta)>0$, because $\beta$ is badly approximable. Now we examine conditions (\ref{m+}) and (\ref{m-}). At first we prove that $\rho$ and $\tau$ with a large enough parameter $t$ satisfy condition (\ref{m+}) for $m_+=1$. Indeed, since $\rho<1/t$ and $\tau<1/t$ (due to (\ref{constr t}) and (\ref{constr tt}), respectively), we have
 % -------------- %
\begin{equation}\label{rho+}
\hspace{-4em}\frac{2\cdot1}{\pi}\tan\left(\frac{\pi}{2}(1\cdot\rho-\lfloor1\cdot\rho\rfloor)\right)=\frac{2}{\pi}\tan\left(\frac{\pi}{2}\rho\right)<\frac{2}{\pi}\tan\left(\frac{\pi}{2t}\right)\to0 \quad\text{as }\;\, t\to\infty
\end{equation}
 % -------------- %
and
 % -------------- %
\begin{equation}\label{tau+}
\hspace{-4em}\frac{2\cdot1}{\pi}\tan\left(\frac{\pi}{2}(1\cdot\tau-\lfloor1\cdot\tau\rfloor)\right)=\frac{2}{\pi}\tan\left(\frac{\pi}{2}\tau\right)<\frac{2}{\pi}\tan\left(\frac{\pi}{2t}\right)\to0 \quad\text{as }\;\, t\to\infty.
\end{equation}
 % -------------- %
Similarly we can show that the number $\varsigma$ for a large enough $t$ satisfies condition (\ref{m-}) with $m_-=1$. Equation~(\ref{constr 1t}) implies $1/(1+1/t)<\varsigma<1$, hence $\lceil\varsigma\rceil=1$ and $1-\varsigma<1/t$; therefore,
 % -------------- %
\begin{equation}\label{varsigma-}
\hspace{-5.5em}\frac{2\cdot1}{\pi}\tan\left(\frac{\pi}{2}(\lceil1\cdot\varsigma\rceil-1\cdot\varsigma)\right)=\frac{2}{\pi}\tan\left(\frac{\pi}{2}(1-\varsigma)\right)<\frac{2}{\pi}\tan\left(\frac{\pi}{2t}\right)\to0 \;\,\text{ as }\;\, t\to\infty.
\end{equation}
 % -------------- %
Finally we prove that $\tau$ with a large enough $t$ obeys condition (\ref{m-}) with the choice $m_-=t$. Since $t/(t+1/t)<t\tau<1$ due to (\ref{constr tt}), we have $\lceil t\tau\rceil=1$ and $1-t\tau<1/(t^2+1)$; hence
 % -------------- %
\begin{equation}\label{tau-}
\hspace{-5.5em}\frac{2t}{\pi}\tan\left(\frac{\pi}{2}(\lceil t\tau\rceil-t\tau)\right)=\frac{2t}{\pi}\tan\left(\frac{\pi}{2}(1-t\tau)\right)<\frac{2t}{\pi}\tan\frac{\pi}{2(t^2+1)}<\frac{2}{\pi}\tan\left(\frac{\pi}{2t}\right).
\end{equation}
 % -------------- %
To sum up, we see from equations~(\ref{rho+})--(\ref{tau-}) that choosing $t$ such that
\begin{equation}\label{condition t}
\frac{2}{\pi}\tan\left(\frac{\pi}{2t}\right)<\mu(\beta)
\end{equation}
guarantees the Bethe--Sommerfeld property of the graph as follows:
 % -------------- %
\begin{itemize}
\item for $a/b=\rho$ and certain repulsive potentials ($\alpha>0$);
\item for $a/b=\varsigma$ and certain attractive potentials ($\alpha<0$);
\item for $a/b=\tau$ and certain potentials of both repulsive ($\alpha>0$) and attractive ($\alpha<0$) type.
\end{itemize}
 % -------------- %
\begin{example}
Let $\beta$ be a root of a quadratic irreducible polynomial over $\mathbb{Z}$ with discriminant $D$. For such $\beta$ we have the estimate $\mu(\beta)\geq\frac{1}{\sqrt{D}}$, which follows from \cite[Sect.~I, Lem.~2E]{Sch80}. Consequently, with regard to (\ref{condition t}), we can define the numbers $\rho,\varsigma,\tau$ by (\ref{constr t})--(\ref{constr tt}) for any $t$ such that $\frac{2}{\pi}\tan\frac{\pi}{2t}<\frac{1}{\sqrt{D}}$.
\end{example}
The idea was applied to construct the number $\theta$ from Theorem~\ref{Thm gap for both alpha}. The continued-fraction representation of $\theta$ from equation (\ref{gap for both alpha}) is $[0;t,t,1,1,1,1,\ldots]$, i.e., $\theta$ was obtained from $\beta=[0;1,1,1,\ldots]=(\sqrt{5}-1)/2$ using scheme (\ref{constr tt}). Since $\mu(\beta)=1/\sqrt{5}$ (because $\beta=\phi^{-1}$, see also Section~\ref{s:golden} and (\ref{Markov rec.})), condition (\ref{condition t}) gives $t\geq3$.

%%%%%%%%%%%%%%%%%%%%%%%%%%%%%%%%%%%%%%%%%%%%%%%%%%%%%%%%%%%%%%%%%%%%%%%%%%%%%%%%%%%%%%%%%%%%%%%%%%%%%%%%%%%%%%%
%%%%---------------------------------------- Concluding remarks -------------------------------------------%%%%
%%%%%%%%%%%%%%%%%%%%%%%%%%%%%%%%%%%%%%%%%%%%%%%%%%%%%%%%%%%%%%%%%%%%%%%%%%%%%%%%%%%%%%%%%%%%%%%%%%%%%%%%%%%%%%%
%  Section
%
\section{Concluding remarks}\label{s:concl}
\setcounter{equation}{0}

Recall first $\Z$-periodic graphs with the period cells linked by more
than a single edge that were briefly mentioned at the end of introduction.
We have not focused on that particular case in our paper; however, a
detailed examination of the spectral structure of $\Z$-periodic graphs in
terms of the Bethe--Sommerfeld behaviour would be interesting, as it may
inspire a new attempt to address the longstanding, still mostly open problem concerning
the Bethe--Sommerfeld conjecture for periodically curved or otherwise
perturbed waveguides.

Secondly, our demonstration that Bethe--Sommerfeld
graphs exist was technical and as such somewhat lacking a simple and
convincing insight. It would be
desirable to achieve a better understanding of the effect.
Let us mention a brief explanation in terms of the
ergodic flow reminiscent of the reasoning used in \cite{BB13}. The
rectangular lattice with Kirchhoff
coupling has no gaps. The set $\Sigma$ of \cite{BB13} covers in that
case the whole torus $[-\pi,\pi)^2$, however, it has `thin points'
corresponding to the quasimomenta values at which the dispersion
curves touch; for definiteness let us focus on the
point $(0,0)$. If we modify the coupling in a way
which is not scale invariant, the flow is perturbed and gaps may
open around these very points. Should there be a finite number of
them, though, the flow has to come close to
$(0,0)$ only rarely and the respective gaps should shrink fast
enough, so that eventually there would be no hits. This heuristic
reasoning allows one to understand why the parameter dependence of
the effect is so tricky. The sketched mechanism of gap opening deserves
to be analyzed rigorously to provide an explanation of the 
Bethe--Sommerfeld property from this point of view, but this task goes 
beyond the scope of the
present paper.

%%%%%%%%%%%%%%%%%%%%%%%%%%%%%%%%%%%%%%%%%%%%%%%%%%%%%%%%%%%%%%%%%%%%%%%%%%%%%%%%%%%%%%%%%%%%%%%%%%%%%%%%%%%%%%%
%%%%--------------------------------------- Acknowledgements ----------------------------------------------%%%%
%%%%%%%%%%%%%%%%%%%%%%%%%%%%%%%%%%%%%%%%%%%%%%%%%%%%%%%%%%%%%%%%%%%%%%%%%%%%%%%%%%%%%%%%%%%%%%%%%%%%%%%%%%%%%%%
\section*{Acknowledgements}
We thank Edita Pelantov\'{a} for a useful discussion and to the referees for their comments that helped us to improve the manuscript. The research was supported by the Czech Science Foundation (GA\v{C}R) within the project 17-01706S.

%%%%%%%%%%%%%%%%%%%%%%%%%%%%%%%%%%%%%%%%%%%%%%%%%%%%%%%%%%%%%%%%%%%%%%%%%%%%%%%%%%%%%%%%%%%%%%%%%%%%%%%%%%%%%%%
%%%%------------------------------------------ Bibliography -----------------------------------------------%%%%
%%%%%%%%%%%%%%%%%%%%%%%%%%%%%%%%%%%%%%%%%%%%%%%%%%%%%%%%%%%%%%%%%%%%%%%%%%%%%%%%%%%%%%%%%%%%%%%%%%%%%%%%%%%%%%%
%\bibliographystyle{acm}
%\bibliography{bibliography}

\end{document}